\documentclass[twocolumn]{autart}

\usepackage{cite}
\usepackage{amsmath,amssymb,amsfonts}
\usepackage{algorithmic}
\usepackage{graphicx}
\usepackage{textcomp}
\usepackage{siunitx}

\usepackage{bm}
\usepackage{array}
\usepackage{xcolor}
\usepackage{tcolorbox}
\usepackage{tikz} 
\usetikzlibrary{backgrounds,arrows,positioning,calc,shapes.geometric,angles,quotes}
\usepackage{todonotes}
\usepackage{wrapfig}
\usepackage{accents}

\DeclareMathOperator*{\rank}{rank}
\def\vp{\varphi}
\DeclareMathSymbol{\mi}{\mathbin}{AMSa}{"39}
\def\pl{\text{+}}
\def\dd{\mathrm{d}}
\newcommand{\lr}[2]{\left\langle{#1},{#2}\right\rangle} 
\newcommand{\mh}[1]{\int_{t-T}^t{#1}\,\dd\tau} 
\def\inv{^{\mi1}}

\def\cL2{\mathop{\mathcal L}_{2}}
\def\cLi{{\mathcal L}_{\infty}}
\def\cLone{{\mathcal L}_{1}}
\def\cD{{\mathcal D}}
\def\cM{{\mathcal M}}
\def\cB{{\mathcal B}}
\def\cE{{\mathcal E}}
\def\cN{{\mathcal N}}
\def\cC{{\mathcal C}}
\def\cO{{\mathcal O}}
\def\cI{{\mathcal I}}
\def\cK{{\mathcal K}}

\newcommand{\ubrace}[2]{{\underbrace{#1}_{#2}}}
\newsavebox{\usqrtbox}
\newcommand{\usqrt}[1]{
  \sbox{\usqrtbox}{%
    \renewcommand{\ubrace}[2]{##1}
    $\displaystyle#1$%
  }%
  \sqrt{\phantom{\usebox{\usqrtbox}}}\hspace*{-\wd\usqrtbox}#1%
}

\newtheorem{remark}{Remark}
\newtheorem{ass}{Assumption}
\newtheorem{lemma}{Lemma}

\newtheorem{theorem}{Theorem}
\newtheorem{defi}{Definition}

\makeatletter%
\@ifclassloaded{autart}
  {
    \DeclareRobustCommand{\qed}{%
      \ifmmode 
      \else \leavevmode\unskip\penalty9999 \hbox{}\nobreak\hfill
      \fi
      \quad\hbox{\qedsymbol}}
    \newcommand{\openbox}{\leavevmode
      \hbox to.77778em{%
      \hfil\vrule
      \vbox to.675em{\hrule width.6em\vfil\hrule}%
      \vrule\hfil}}
    \newcommand{\qedsymbol}{\openbox}
    \newenvironment{proof}[1][\proofname]{\par
      \normalfont
      \topsep6\p@\@plus6\p@ \trivlist
      \item[\hskip\labelsep\itshape
        #1.]\ignorespaces
    }{%
      \qed\endtrivlist
    }
    \newcommand{\proofname}{Proof}
    \renewcommand{\theenumi}{\textit{\roman{enumi}}}
  }%
  {
    \renewcommand{\theenumi}{(\textit{\roman{enumi}}}
  }%
\makeatother

\begin{document}
\begin{frontmatter}
\title{Modulating function based algebraic observer coupled with stable output predictor for LTV and sampled-data systems}
\author[TUI]{Matti Noack}\ead{matti.noack@tu-ilmenau.de},
\author[KAUST]{Ibrahima N'Doye}\ead{ibrahima.ndoye@kaust.edu.sa},
\author[TUI]{Johann Reger}\ead{johann.reger@tu-ilmenau.de},
\author[INRIA]{Taous-Meriem~Laleg-Kirati}\ead{taous-meriem.laleg@inria.fr}

\address[TUI]{Technische Universit\"at Ilmenau, Ehrenbergstra{\ss}e 29,
D-98693 Ilmenau}
\address[KAUST]{King Abdullah University~of~Science~and~Technology,~Thuwal 23955-6900,~Saudi~Arabia}
\address[INRIA]{National Institute for Research in Digital Science and Technology, INRIA, Paris-Saclay, 91120 Palaiseau,~France}


\begin{abstract}
	This paper proposes an algebraic observer-based modulating function approach for linear time-variant systems and a class of nonlinear systems with discrete measurements. The underlying idea lies in constructing an observability transformation that infers some properties of the modulating function approach for designing such algebraic observers. First, we investigate the algebraic observer design for linear time-variant systems under an observable canonical form for continuous-time measurements. Then, we  provide the convergence of the observation error in an $\cL2$-gain stability sense. Next, we develop an exponentially stable sampled-data observer which relies on the design of the algebraic observer and an output predictor to achieve state estimation from available measurements and under small inter-sampling periods. Using a trajectory-based approach, we prove the convergence of the observation error within a convergence rate that can be adjusted through the fixed time-horizon length of the modulating function and the upper bound of the sampling period. Furthermore, robustness of the sampled-data algebraic observer, which yields input-to-state stability, is inherited by the modulating kernel and the closed-loop output predictor design. Finally, we discuss the implementation procedure of the MF-based observer realization, demonstrate the applicability of the algebraic observer, and illustrate its performance through two examples given by linear time-invariant and linear time-variant systems with nonlinear input-output injection terms.
\end{abstract}

\end{frontmatter}



\section{Introduction}
The subject of designing non-asymptotic state estimators remains an active area of research, especially considering the extension of suitable system classes while taking different sources of uncertainty into account. This group of fixed-time convergent observers encompasses techniques like sliding-mode and super twisting observers \cite{SEFL:14}, time-dependent scaling and delay-based approaches \cite{MazencMN20} as well as the modulating function method \cite{RegerJ15}, the latter of which is emphasized in this work. Major advantages incorporate arbitrarily tunable convergence time when compared to only exponentially converging algorithms, the related robustness and disturbance characterizations.
One of the major challenges in the given context is dealing with discrete and unevenly sampled measurements since the mentioned approaches focus on a continuous-time perspective. Furthermore, the considered family of systems is of linear time-varying (LTV) nature with the addition of measurement-based nonlinearities allowing for a large set of relevant technical problems.\\
LTV systems are important models for several engineering applications \cite{MMN:14} and arise mainly with approximations of nonlinear systems
and nonlinear tracking control problems \cite{MaP:00}. Many efforts have been devoted to finding less conservative stability conditions of LTV systems through Lyapunov function and perturbation analysis \cite{Zho:16}. However, the stability analysis of LTV systems, which is one of the first open problems in control theory \cite{AeP:99}, remains challenging. 
Especially for the design of non-asymptotic observers in the LTV case, particular structural properties are often required \cite{DTF:22}.\\
Sampled-data setups, also arising in the context of continuous-discrete systems and hybrid systems, in which the output measurements are available at discrete-time instants, appear frequently in networked control systems \cite{HNX:07}, cyber-physical systems \cite{San:16}, and manufacturing \cite{Chr:13}. For such a class of problems that comprises continuous-time dynamics with discrete output measurements, the corresponding exact discrete-time system is challenging to compute for nonlinear systems \cite{NABLH:14}. In such context, continuous-discrete observer designs have been proposed to guarantee the observation error convergence for a sufficiently small sampling period in \cite{NABLH:14} and for any constant sampling period in \cite{SuP:97}. These observers have been extended to the case of sporadic measurements (see, for instance, \cite{RaA:07,RKA:08,DANS:15,EHEP:17,STZ:19}). Another perspective on the sampled-data case can be taken via the switched systems approach. For such particular class of systems, several works have been devoted to designing observers capable of estimating the switching signal (see, for instance, \cite{Bat:13,BBBS:13,ZVPR:19}).
Further investigations focus on constructing conditions regarding the sampling periods for the sampled-data case to ensure asymptotic or exponential stability guarantees of the observation error as in \cite{BeS:22}, where the state exhibits instantaneous jumps.
Compared to such continuous-discrete observers including switched observers, the idea of compensating the sampling effect with inter-sample output predictors that reinitialize the observer at the sampling times provides essential features. This includes exponential convergence, robustness features, and ease of implementation~\cite{KAG:20}. Sampled-data observer designs based on the so-called dynamic extension via output prediction have been initially developed for detectable linear and triangular nonlinear systems \cite{KaK:09}, and have been the subject of several extensions~\cite{FAGBL:16,AGAE:21,KAG:20,MazencMN20,MazencMN22}.
In light of the above review, the problem of designing robust state estimation algorithms for continuous-time dynamic systems in which the output measurements are sampled at specific time instances poses fundamental challenges due to the sampling nature of the output and discrete noisy measurements. Therefore, sufficient conditions on the sampling time are needed to guarantee convergence of the observation error.\\
In line with this, as a burst of research activity in estimation problems, the modulating function method (MFM) offers several advantages, including a systematic way for state reconstruction within finite time, not depending on the initial conditions. Additionally, modulating function (MF) based methods provide robustness that is inherited from the integration process of the modulating kernel, which can be interpreted as a finite impulse response (FIR)-filter \cite{LiuGP11}. Initially, MF methods have been developed for parameter identification of continuous-time linear ordinary differential equations (ODEs) \cite{Shin:54,PrR:93}. After the introduction of the algebraic formulation of exact state observers \cite{Byrski03} for general linear time-invariant (LTI) problems, MFs have further been extended to joint state and parameter estimation for linear systems \cite{LiL:15} and systems in observability normal form (ONF) with an additional polynomial nonlinearity structure \cite{RegerJ15} or a general nonlinear input-output injection in \cite{LiuLBWW22}.
Recently, MF techniques have been extended towards different classes of nonlinear systems \cite{NoackKR22,MNL:23,ANL:23}. These works leverage the non-asymptotic state estimation concepts and establish a significant step towards accommodating for nonlinearities in ODEs by using MFs.\\
One of the key aspects of this work in terms of dealing with the time-variant setup is the transformation into observability companion form (OCF) to take advantage of the differential structure of the related input-dynamics to which the MFM can directly be applied. This includes the classical concept of strong observability and the existence of a Lyapunov transformation (matrix function) preserving the stability properties of the original system.
Subsequently, the input-output nonlinearity is incorporated by treating it as a multi-dimensional known input signal.
The last step of dealing with corrupted measurements at unknown discrete time instances and verifying stability for the sampled-data case is performed by following the strategy established in \cite{MazencMN22} consisting of the dynamic prediction of the unknown output as well as of applying the trajectory-based approach \cite{MazencM15} which turns out to be especially suitable for system formulations with delays or finite history integrals.
In comparison to the aforementioned work, the main distinction here lies within the non-asymptotic observer approach, which in \cite{MazencMN22} is based on temporal shifts with respect to a selected time step and results in a specific delay-based observability characterization. The corresponding transformation matrix is obtained by computing the transition matrix, whereas here, differentiation of the matrix functions has to be performed to obtain the observer-related coordinates. This parallel consideration of the two particular perspectives is further discussed throughout this paper.
The main advantages from deploying a MFM observer structure are the tuning opportunities that the kernel-based approach offers. First, the moving horizon length can again be seen analogously to the temporal data shift approach and directly determines the convergence time. Secondly, the kernel shape itself can be chosen according to the weighted disturbance gains and also with regards to the stability margin originating from the stability analysis.\\
This paper's result encompasses the following novel developments:
\vspace{-1ex}
\begin{enumerate}
    \item We propose a MF-based observer
    for multi-input LTV systems exploiting the properties of the output dynamics related to the observability companion form (see Section \ref{sec:algobs}).
    \item We provide a compact $\cL2$ error characterization for the exact observer with respect to measurement and process disturbances (see Theorem \ref{thm:mf_observer}).
    \item We extend the MF-based state reconstruction to sampled-data systems with input-output nonlinear terms by dynamic extension using an output predictor (see Section \ref{sec:sd_estimation})
    while proving a stability criterion for exponential convergence of the state estimation error within its ISS bounds using the trajectory-based approach from \cite{MazencM15} (see Theorem \ref{thm:sd_observer}).
\end{enumerate}
\vspace{-1ex}
Finally, we discuss the implementation procedure of the MF-based observer realization, demonstrate the applicability of the algebraic observer and illustrate its performance through two examples given by LTI and LTV systems with nonlinear input-output injection terms. Thereby, the conservativeness of the error bounds is assessed during a simulative comparison study.\\
The paper is organized as follows: In Section \ref{sec:algobs}, we study the algebraic observer algorithm for linear time-varying systems resorting to an observable canonical form for continuous-time systems. We then show the convergence of the observation error in $\cL2$-gain sense through direct implicit inequalities. In Section \ref{sec:sd_estimation}, we extend the methodology to the sampled-data case and provide the stability result of the sampled-data observation error and its proof, guaranteeing exponential input-to-state stability of the estimation error in the presence of process disturbances and measurement noise. In Section \ref{sec:implementation}, we proceed with guidelines on the numerical implementation of the estimation scheme which is demonstrated using two examples. These examples illustrate the observer performance and let compare it to a similar approach. Finally, conclusions and perspectives are stated in Section~\ref{concl}.\\[1ex]
{\bf Notation:}
Matrix $A^{\top}$ is the transposed matrix of $A$ and the identity matrix is $I$. The $j$-th unit vector is written as $e_j$, the exchange matrix $J_j$ of dimension $j\times j$ only contains ones on the anti-diagonal and zeros elsewhere.
The short form $\|\cdot\|=\|\cdot\|_2$ denotes the Euclidean vector and induced matrix norm.
For a compact interval $\mathcal{I}$ and a measurable function $\xi$, the $\mathcal{L}_p$ norm is written as $\|\xi\|_{\mathcal{L}_p(\mathcal{I})}=\left(\int_\mathcal{I}\|\xi(\tau)\|^p\dd\tau\right)^{\frac{1}{p}}$.
$\cL2(\mathcal{I})$ represents a Hilbert space with inner product $\lr{\xi}{\eta}_\mathcal{I}=\int_\cI\xi(\tau)^\top\!\eta(\tau)\dd\tau$ for $\xi,\eta\in\cL2(\mathcal{I})$.
The space of signals with finite $\cL2$ norm on every compact interval is called $\cL2^\text{loc}$.
With $\dd_t=\frac{\dd}{\dd t}$, the operator $\cD_k=[\dd_t^{k\mi1},\ldots,\dd_t,\text{id}]$ denotes the vector differentiation (as row vector) up to order $k-1\in\mathbb{N}$.
Its adjoint operator is defined as $\cD_k^\ast=[(\mi1)^{k\mi1}\dd_t^{k\mi1},\ldots,\mi\dd_t,\text{id}]^\top$ (column vector).



\section{Algebraic Observer Design}
\label{sec:algobs}
Before extending the observer approach to sampled-data systems with additive output nonlinearity, a linear time-variant system is studied to establish the basic algorithm and to characterize the impact of different disturbances.

\subsection{Linear Time-variant Problem Setup}
Consider the perturbed multi-input linear time-variant (LTV) system
\begin{equation}\label{eq:sys}
 \left\{\begin{array}{l}
        \dot{x}=A(t)x+B(t)u+E(t)d \\
        \tilde{y}=\textcolor{black}{C(t)}x+\nu=y+\nu
    \end{array}\right.
\end{equation}
with bounded and sufficiently smooth matrix-valued functions $A\in\cC^n(\mathbb{R}^+_0,\mathbb{R}^{n\times n}), B\in\cC^n(\mathbb{R}^+_0,\mathbb{R}^{n\times m}),C\in\cC^n(\mathbb{R}^+_0,\mathbb{R}^{1\times n})$ and $E\in\cC^n(\mathbb{R}^+_0,\mathbb{R}^{n\times p})$
where $x(t)\in\mathbb{R}^n, y(t), \tilde{y}(t)\in\mathbb{R}, u(t)\in\mathbb{R}^m,$ $d(t)\in\mathbb{R}^p,\ \nu(t)\in\mathbb{R}$ are the unknown state, associated output, measurement output, control input, process disturbance and sensor noise, respectively.
We assume disturbances and noise as $d,\nu\in\cL2^{\text{loc}}$, i.e. with finite energy on every compact time interval. Additionally, the following structural assumption is central for the following argumentation.
\begin{ass}\label{as:obsvb}
Let system representation \eqref{eq:sys} be uniformly observable with respect to time \cite{SilvermanMeadows67} meaning $\rank(\cO(t))=n\,\forall t\geqslant0$ with observability matrix defined as
\vspace{-2ex}
\begin{equation}\label{eq:obsmat}
    \cO(t)\!=\!\begin{pmatrix}
        \cN^0C \\ \cN^1C \\ \vdots \\ \cN^{n\mi1}C
    \end{pmatrix}
    \text{ where }
    \begin{cases}
        \cN^0C=C(t) \\
        \cN^kC=\frac{\dd}{\dd t}\cN^{k\mi1}C\\
        \hspace{5ex}+(\cN^{k\mi1}C)A(t)
    \end{cases}\!\!\!\!.
\end{equation}
Furthermore, it is assumed to be strongly observable \cite[Lem. 2.32]{TsakalisIoannou93}, namely $\exists\,\delta>0:|\det(\cO(t))|>\delta\,\forall t\geqslant0$.
\end{ass}

Assumption \ref{as:obsvb} guarantees the existence of a Lyapunov transformation $P\in\cC^1(\mathbb{R}_0^+,\mathbb{R}^{n\times n})$, being uniformly regular and bounded with bounded derivative $\dot{P}$, which is constructed by using \eqref{eq:obsmat} with $q(t)=\cO(t)\inv e_n$:
\begin{equation}\label{eq:trafo}
    P(t)\!=\!\left[q|\cK q|\cdots|\cK^{n\mi1} q\right]\!J_n
    \text{ where }
    \begin{cases}
        \cK^0q\!=\!q(t) \\
        \cK^kq\!=\!\mi\frac{\dd}{\dd t}\cK^{k\mi1}q\\
        \hspace{2ex}+A(t)(\cK^{k\mi1}q)
    \end{cases}
\end{equation}
transferring system \eqref{eq:sys} into so-called canonical (phase variable) observability form via the coordinate change $z=P(t)^{\mi1}x$. Note that the anti-diagonal exchange matrix $J_n$ leads to a permutation of the transformed state such that $y=z_1$. It also allows for a shift from the state space perspective to the input-output behavior being obtained from this observability companion form (OCF) as discussed in \cite{krener_linearization_1983} for nonlinear systems.
Consequently, the following structure used throughout this section turns out to be reminiscent of the linear time-invariant (LTI) form used in \cite[Th.1]{RegerJ15} for MF-based observer design:
\begin{equation}\label{eq:OCF}
    \begingroup
     \renewcommand*{\arraystretch}{1.1} 
     \arraycolsep=2.5pt 
     \left\{
     \begin{array}{l}
        \!\!\dot{z}\!=\!\!\underbrace{\begin{pmatrix}
            \mi a_{n\mi1}(t)&1&&\\\mi a_{n\mi2}(t)&0&\ddots&\\\vdots&\vdots&\ddots&1\\\mi a_0(t)&0&\cdots&0
        \end{pmatrix}}_{=A_o(t)}\!\!z\!+\!\!\underbrace{\begin{pmatrix}
            B_{n\mi1}^\top(t)\\ B_{n\mi2}^\top(t)\\\vdots\\ B_0^\top\!(t) \end{pmatrix}}_{=B_o(t)}\!\!u\!+\!\!\underbrace{\begin{pmatrix}
            E_{n\mi1}^\top(t)\\ E_{n\mi2}^\top(t)\\\vdots\\ E_0^\top\!(t)
        \end{pmatrix}}_{=E_o(t)}\!\!d \\
        \!\tilde{y}\!=\!\textcolor{black}{C_o}z+\nu=z_1+\nu
    \end{array}
    \right.
    \endgroup
\end{equation}
with transformed matrices $A_o(t)=P^{\mi1}(AP-\dot P)\in\mathbb{R}^{n\times n}, B_o(t)=P^{\mi1}B\in\mathbb{R}^{n\times m},\textcolor{black}{C_o=CP=e_1^\top}\in\mathbb{R}^{1\times n}$ and $E_o(t)=P^{\mi1}E\in\mathbb{R}^{n\times p}$ by using \eqref{eq:trafo}.

The goal is to algebraically reconstruct the transformed state $z$ in fixed time, based on $u$ and measured $y$ while being able to characterize the impact of unknowns $d$ and~$\nu$.

\begin{remark}[Observability canonical form]\label{rem:obs_canon}
An alternative system representation is the observability canonical form for time-variant systems as discussed in \cite[Sec.3]{zeitz_observability_1984}. For construction, the observability matrix $\cO$ is directly used as the transformation matrix. In the present article, this representation is not used as it leads to a different input-output representation that is less suitable for the MF application.
This comes down to the distinction between left and right Polynomial Differential Operators (PDO) \cite[Def.2.1\&2.2]{TsakalisIoannou93}.
In \cite{NoackKR22}, such a form is utilized for designing an MFM-based observer for time-variant constellations resulting from linearization. 
\end{remark}

\subsection{Modulating Function Method}
In the context of MF based algebraic estimation methods, often a moving horizon perspective is emphasized focusing on a finite set of past measurement data.
This generally facilitates the analysis of stability and the effect of disturbance signals on the time window. For the construction of an appropriate modulation kernel with regards to a particular estimation problem, the boundary conditions play a vital role in extracting desired information.
The following definition combines the established structures from \cite[Def.2]{RegerJ15} and \cite[Def.3]{NoackKR22}, and extends them to a vector kernel case.
\vspace{0.5ex}
\begin{defi}[Left \& unitary MFs]\label{def:mf}
    For a fixed time horizon length $T>0$, the vector-valued function $\varphi\in\mathcal{C}^k([0,T],\mathbb{R}^k),$ is called left modulating function (LMF) of order $k\in\mathbb{N}$ if
    \[
        \varphi^{(i)}(0)=0,\,\varphi^{(i)}(T)=c_i\in\mathbb{R}^k
    \]
    $\forall i\in\{0,1,\ldots,k-1\}$.
    Furthermore, it is called unitary modulating function (UMF) if
    \begin{equation}\label{eq:umfcond}
        \varphi^{(i)}(T)=c_i=(-1)^i\,e_{k-i}
    \end{equation}
    with $e_j\in\mathbb{R}^k$ the $j$-th unit vector.
    The corresponding modulation operator $\mathcal{M}:\cL2^\text{loc}\rightarrow\mathbb{R}^k$ forms an inner product applied to the signal $\xi\in\cL2^\text{loc}([T,\infty),\mathbb{R})$ as per
    \begin{equation}\label{eq:modop}
        \mathcal{M}[\xi]=\mh{\varphi(\tau-t+T)\,\xi(\tau)}\,
        =\lr{\varphi^\top\!\!}{\xi}
    \end{equation}
   where the modulation takes place vector component-wise over a moving horizon.\\
\end{defi}


Now, consider the representation of the input-output dynamics of system \eqref{eq:sys}, corresponding to the time-variant OCF \eqref{eq:OCF}, that is
\begin{equation}\label{eq:inout}
     y^{(n)}+\sum_{i=0}^{n-1} \big(a_i(t)y\big)^{(i)}=
     \sum_{j=0}^{n-1} \big(B_j^\top\!(t)u\big)^{(j)}+\sum_{l=0}^{n-1}\big(E_l^\top\!(t)d\big)^{(l)}
\end{equation}
where $a(t)=-A_oe_1=[a_{n-1},\ldots,a_0]^\top\in\mathbb{R}^n$ holds the characteristic coefficients related to the right PDO for $A$ and $A_o$ with corresponding input vectors $B_j(t)\in\mathbb{R}^m,E_l(t)\in\mathbb{R}^p$.
Note that in case of the companion form as a particular normal form, the differential operators are outside of the product of signals and coefficient functions resulting in a right PDO formulation suitable for estimation.\\
For applying the MFM, input-output equation \eqref{eq:inout} is reformulated into a linear operator representation analog to the framework established in \cite{NoackKR22}:
\begin{equation}\label{eq:inoutop}
    \underbrace{y^{(n)}+\cD\big[a(t)\,y\big]}_{L y}=\underbrace{\cD\big[B_o(t)\,u\big]}_{\cB u}+\underbrace
    {\cD\big[E_o(t)\,d\big]}_{\cE d}
\end{equation}
with linear time-variant differential operator $L(t)[\,\cdot\,]=\big(\frac{\dd^{n}}{\dd t^{n}}[\,\cdot\,]+\cD (a(t)[\,\cdot\,])\big)$ and input operators $\cB(t)[\,\cdot\,]=\cD (B_o(t)[\,\cdot\,]), \cE(t)[\,\cdot\,]=\cD (E_o(t)[\,\cdot\,])$ constructed via the short form differentiation operator $\cD=\cD_n$. The time argument is occasionally omitted for simplicity.
The application of the MFM to the linear structure from \cite[Th.1]{NoackKR22}, is adapted similarly in the following lemma. Specifically, the case with time-variant coefficients inside the differentiation due to the OCF is considered analogously while generalizing the result to systems with multiple inputs.
\begin{lemma}[Adjoint equation \& boundary values]\label{lem:shift}
    For a given LMF $\varphi:[0,T]\rightarrow\mathbb{R}^n$ of order $n$, applying the modulation operator $\cM$ from \eqref{eq:modop} to \eqref{eq:inoutop} leads to
    \begin{equation}\label{eq:shift_bc}
        \begin{aligned}
            \lr{L^\ast\varphi}{y}+\Delta\,Y(t) 
            =\lr{\cB^\ast\!\vp}{u}+\sum_{j=1}^{m}\,\Gamma_\cB^j\,U_j(t)\hspace{6ex} \\[-1.5ex]
            +\lr{\cE^\ast\!\vp}{d\,}+\sum_{l=1}^{p}\,\Gamma_\cE^l\,D_l(t)
        \end{aligned}
    \end{equation}
    with adjoint operators
    \begin{equation*}
        \begin{aligned}
            L^\ast(t)[\,\cdot\,]=(-1)^n\frac{\dd^{n}}{\dd t^{n}}[\,\cdot\,]^\top\!+a(t)^\top\cD^\ast[\,\cdot\,]^\top\,,\\
            \cB^\ast(t)[\,\cdot\,]=B_o^\top\!(t)\cD^\ast[\,\cdot\,]^\top\,,\,\,
            \cE^\ast(t)[\,\cdot\,]=E_o^\top\!(t)\cD^\ast[\,\cdot\,]^\top\,.
        \end{aligned}
    \end{equation*}
    Therein, the collection of differentiation signal vectors $Y=[y,\dot{y},\ldots,y^{(n-1)}]^\top,$ $U_j=[u_j,\dot{u}_j,\ldots,u_j^{(n-1)}]^\top$ and $D_l=[d_l,\dot{d}_l,\ldots,d_l^{(n-1)}]^\top$ are separated via the boundary condition coefficient matrices
    \[
        \begin{aligned}
            \Delta =&\, \big[L_1^\ast\varphi(T)|\cdots|L_n^\ast\varphi(T)\big] \,, \\
            \Gamma_\cB^j =&\, \big[\cB_{j,1}^\ast\varphi(T)|\cdots|\cB_{j,n}^\ast\varphi(T)\big] \,, \\
            \Gamma_\cE^l =&\, \big[\cE_{l,1}^\ast\varphi(T)|\cdots|\cE_{l,n}^\ast\varphi(T)\big] \,.
        \end{aligned}
    \]
    The minor adjoint operators $L_k^\ast,\cB_{j,k}^\ast$ and $\cE_{l,k}^\ast$ are analogously associated with
    \begin{equation*}
        \begin{aligned}
            &L_k\!\!=\!\!\!\sum_{i=0}^{n\mi k}\mathrm{d}^{n\mi k\mi i}_t\big(a_{n\mi i}(t)[\,\cdot\,]\big),
            \cB_{j,k}\!\!=\!\!\!\sum_{i=0}^{n\mi k}\mathrm{d}^{n\mi k\mi i}_t\big(B_{j,n\mi i}(t)[\,\cdot\,]\big)\,, \\[-0.5ex]
            &\hspace{2cm}\cE_{l,k}\!=\!\sum_{i=0}^{n-k}\mathrm{d}^{n\mi k\mi i}_t\big(E_{l,n\mi i}(t)[\,\cdot\,]\big)
        \end{aligned} 
    \end{equation*}
    where $L_0=L,L_n=1,\cB_{j,0}=\cB_j,\cE_{l,0}=\cE_l$ and $\cB_{j,n}=\cE_{l,n}=0$ particularly hold.
\end{lemma}
\begin{proof}
    In \cite[proof of Th.1]{NoackKR22}, this argumentation is realized for the SISO LTI case.
    Analog to that derivation, partial integration of the modulated operator equation \eqref{eq:inoutop} representing the output dynamics results in the shift of the differential operations in adjoint form to the kernel.
    Accordingly, the corresponding boundary values from LMF condition $\vp^{(i)}(T)=c_i$ lead to the isolation of the signal differentiation vectors $Y,U_j,D_l$ as a rearranged form of Green's formula.
    Here, the derivative shift related to the vector-valued input signals $u$ and $d$, is treated similar to the scalar operator case.
    As the time-variant coefficient functions are situated inside the the differentiation, the adjoint shift directly effects the modulation kernel with coefficient derivatives only appearing in the boundary isolation step.
\end{proof}

The direct modulation result \eqref{eq:shift_bc} illustrates the derivative shift property of the MF approach, moving the differentiation operator from the system variables to the modulation kernel. A challenge for the state estimation task is the appearance of the input signal boundary conditions including their generally unknown derivative values.

\subsection{MF-based Observer Algorithm}
In order to have the input boundary values not interfere with the estimation process, the properties of the OCF are exploited. The algebraic observer for reconstructing the OCF state of system \eqref{eq:OCF}, while quantifying the disturbance impact, is based on the following assumption.
\vspace{0.5ex}
\begin{ass}\label{as:disturb}
    The process disturbance $d$ and the measurement perturbation $\nu$ possess a global bound on their local $\cL2$ norm regarding the moving horizon $[t-T,t]$:
    \begin{enumerate}
        \item $\|\nu\|_{\cL2[t-T,t]}\leqslant\bar{\nu}=\underset{t\geqslant0}{\sup}\,\|\nu\|_{\cL2[t-T,t]}$\,,
        \item $\|d\|_{\cL2[t-T,t]}\leqslant\bar{d}=\underset{t\geqslant0}{\sup}\,\|d\|_{\cL2[t-T,t]}$\,.
    \end{enumerate}
\end{ass}
\vspace{0.5ex}
Using the operator form from Lemma \ref{lem:shift} and rearranging the boundary conditions in the context of the ONF enables the development of an MF-based state estimator with error characterization as stated in following fundamental theorem.
\begin{theorem}[MF-based observer]\label{thm:mf_observer}
    For a given UMF $\varphi:[0,T]\rightarrow\mathbb{R}^n$ of order $n$ and disturbances $d,\nu$ fulfilling Assumption \ref{as:disturb}, the state of system \eqref{eq:OCF} is reconstructed by
    \begin{equation}\label{eq:mfobs}
        \hat{z}=-\lr{L^\ast\varphi}{\tilde{y}}+\lr{\cB^\ast\vp}{u}
    \end{equation}
    with bound on the error $e_z(t)=z-\hat{z}$ such that $\forall t\geqslant t_0+T$
    \begin{equation*}
        \|e_z(t)\|_2
        \leqslant\|\cE^\ast\varphi\|_{\cL2[0,T]}\,\bar{d}
        +\|L^\ast\varphi\|_{\cL2[0,T]}\,\bar{\nu}\,.
    \end{equation*}
\end{theorem}
\begin{proof}
    Applying UMF $\vp$ to the equivalent input-output form \eqref{eq:inoutop}, instead of their shape in \eqref{eq:shift_bc} resulting from the LMF, the insertion of the specific unit-vector-like boundary conditions leads to the coefficient matrices being simplified to the following Toeplitz structures:
    \renewcommand*{\arraystretch}{1}
    \[
        \begin{aligned}
            \Delta(t) = \begin{pmatrix}
                1&0&\cdots&0\\ a_{n\mi1}(t)&1&\ddots&\vdots\\\vdots&\ddots&\ddots&0\\ a_1(t)&\cdots&a_{n\mi1}(t)&1
            \end{pmatrix}\,, \hspace{2cm} \\[0.5ex]
            \begingroup
            \arraycolsep=1pt
                \Gamma_\cB^j(t) \!\!=\!\! \begin{pmatrix}
                0&\cdots&0\\ B_{j,n\mi1}&\ddots&\vdots\\\vdots&\ddots&0\\ B_{j,1}&\cdots&B_{j,n\mi1}
            \end{pmatrix}\!,
                \Gamma_\cE^j(t) \!\!=\!\! \begin{pmatrix}
                0&\cdots&0\\ E_{j,n\mi1}&\ddots&\vdots\\\vdots&\ddots&0\\ E_{j,1}&\cdots&E_{j,n\mi1}
            \end{pmatrix}\!.
            \endgroup
        \end{aligned}
    \]
    This corresponds to the classical differential parameterization of the ONF state
    \[
        z=\Delta(t) Y-\sum_{j=1}^{m}\,\Gamma_\cB^j(t)\,U_j(t)-\sum_{l=1}^{p}\,\Gamma_\cE^l(t)\,D_l(t)\,.
    \]
    The shift by partial integration can then be sorted in the opposite way as demonstrated in Lemma \ref{lem:shift} by separating the BCs with respect to the MF derivatives, that is
    \begin{align}
        &\lr{\vp^\top\!}{Ly-\cB u-\cE d}
        =\lr{L^\ast\vp}{y}-\lr{\cB^\ast\vp}{u}-\lr{\cE^\ast\vp}{d} \nonumber\\
        &
        \!\!+\!\!\left[\left((\mi1)^{n\mi1}\vp^{(n\mi1)}|\cdots|\!\mi\!\dot\vp|\vp\right)\!\!\!\!
        \vphantom{\begin{pmatrix} y\\\vdots\\y\end{pmatrix}}\right.
        \underbrace{\begin{pmatrix}
            L_ny \\ L_{n\mi1}y\!\mi\!\cB_{n\mi1}u\!\mi\!\cE_{n\mi1}d \\ \vdots \\ L_1y\!\mi\!\cB_1u\!\mi\!\cE_1d
        \end{pmatrix}}_{=z(\tau-t+T)}\!
        \left.\vphantom{\begin{pmatrix} y\\\vdots\\y\end{pmatrix}}\!\right]^T_0\nonumber\\
        &\overset{\eqref{eq:umfcond}}{=}\lr{L^\ast\vp}{y}-\lr{\cB^\ast\vp}{u}-\lr{\cE^\ast\vp}{d}+z=0\,. \label{eq:state_mod_form}
    \end{align}
    This allows for expressing the unknown state $z$ in terms of modulated input and output signals without need to consider the boundary terms from \eqref{eq:shift_bc}.
    Without knowledge on the disturbance $d$ and only resorting to the noisy measurement signal $\tilde{y}=y+\nu$, the state estimate is selected as stated in \eqref{eq:mfobs}.
    This choice results in the error equation
    \[
        \begin{aligned}
            e_z=&\,z-\hat{z}=-\lr{L^\ast\vp}{y}+\lr{\cB^\ast\vp}{u}+\lr{\cE^\ast\vp}{d} \\
        &\hspace{3.5cm}+\lr{L^\ast\vp}{\tilde{y}}-\lr{\cB^\ast\vp}{u} \\
        =&\,-\lr{L^\ast\vp}{y}+\lr{L^\ast\vp}{\tilde{y}}+\lr{\cE^\ast\vp}{d}\\
        =&\,\lr{L^\ast\vp}{\nu}+\lr{\cE^\ast\vp}{d}\,.
        \end{aligned}
    \]
    Thus, the error norm is bounded accordingly:
    \begingroup
    \allowdisplaybreaks
    \begin{align*}
        &\|e_z\|_2\leqslant\left\|\mh{\cE^\ast(\tau)[\vp](\sigma)\,d(\tau)}\right\|_2  \\[-1.0ex]
        &\hspace{2.8cm}+\left\|\mh{L^\ast(\tau)[\vp](\sigma)\,\nu(\tau)}\right\|_2 \\[-1.0ex]
        &\leqslant\mh{\left\|\cE^\ast(\tau)[\vp](\sigma)\,d(\tau)\right\|} \\[-1.0ex]
        &\hspace{2.8cm}+\mh{\left\|L^\ast(\tau)[\vp](\sigma)\,\nu(\tau)\right\|}\\
        &\overset{\text{H\"older}}{\leqslant}\!\!\sqrt{\int_0^T\left\|\cE^\ast(\sigma\pl t\!\mi\! T)[\vp](\sigma)\right\|^2\mathrm{d}\sigma}\,\,\|d\|_{\cL2[t-T,t]} \\[-1.0ex]&\hspace{1.1cm}+\sqrt{\int_0^T\left\|L^\ast(\sigma\pl t\!\mi\! T)[\vp](\sigma)\right\|^2\mathrm{d}\sigma}\,\|\nu\|_{\cL2[t-T,t]} \\
        &\overset{\text{As.\ref{as:disturb}}}{\leqslant}\|\cE^\ast(\cdot\,\pl t\!\mi\! T)\varphi\|_{\cL2[0,T]}\,\bar{d}+\|L^\ast(\cdot\,\pl t\!\mi\! T)\varphi\|_{\cL2[0,T]}\,\bar{\nu}
    \end{align*}
    \endgroup
    with the argument substitution $\sigma=\tau-t+T$.
\end{proof}
\vspace{0.5ex}
This theorem demonstrates the straightforward analysis for such an algebraic scheme on a finite time horizon.
In \cite{Byrski03}, it is shown that all algebraic estimators for LTI systems take the form of inner products of particular kernel with the input and output signals which is consistent with \textcolor{black}{\eqref{eq:mfobs}}.
Compared to the exact state observer there, the stated theorem provides a simpler structure for calculating the kernels as well as for quantifying the influence of uncertainties.
\vspace{0.5ex}
\begin{remark}[Transformation of observable system]\label{rem:obs_trafo} To apply the observer scheme in Theorem \ref{thm:mf_observer} to the more general observable system \eqref{eq:sys}, the transformation to time-variant observability companion form has to be performed in the $x$-coordinates. Accordingly, the transformation $P$ has to be included in the observer calculation \eqref{eq:mfobs} as per
    \[
        \hat{x}=P(t)\hat{z}=P(t)\big(-\lr{L^\ast\varphi}{\tilde{y}}+\lr{\cB^\ast\varphi}{u}\big)
    \]
    resulting in the adapted error bound $e_x=x-\hat{x}=P(t)e_z$.
    The involved Lie derivatives rely on sufficiently smooth matrix functions $A(\cdot),B(\cdot),C(\cdot)$ and feature a similar structure for systems with output nonlinearities.
\end{remark}

Despite the property of the algebraic observers to obtain an estimate in time prescribed by the horizon length, the analysis of dynamic systems with cross-couplings to the MF estimator turns out to be challenging.
The $\cL2$ bound investigation constitutes an important step towards characterizing the stability of the error system propagation with a more comprehensive architecture. It prepares the convergence or contraction analysis as well as subsequent robust design methods.

\section{Sampled-data Estimation via Output Predictor Closed-loop}
\label{sec:sd_estimation}

The class of sampled-data systems with additive output nonlinearaties is inspired by the time-invariant problem in \cite{MazencMN20} and the time-variant extension in the subsequent contribution \cite{MazencMN22}. The authors also considered disturbances besides other advanced aspects.
Throughout this work, the state estimation problem is addressed for the sampled-data systems of the form
\begin{equation}
    \label{eq:sd_sys}
 \left\{\begin{array}{l}
        \dot{x}=A(t)x+\phi(y,u,t)+Ed \\
        \tilde{y}=Cx(t_i)+\nu=y(t_i)+\nu\quad\forall t\in[t_i,t_{i+1})
    \end{array}\right.
\end{equation}
with bounded dynamic matrix-valued function of time, $A\in\cC^n(\mathbb{R}^+_0,\mathbb{R}^{n\times n})$, constant output matrix $C\in\mathbb{R}^{1\times n}$ and disturbance matrix $E\in\mathbb{R}^{n\times p}$ as well as the bounded continuous input-output nonlinearity $\phi:\mathbb{R}\times\mathbb{R}^m\times\mathbb{R}^+_0\rightarrow\mathbb{R}^n$.
The sampling instants $t_i,i\in\mathbb{N}$, are potentially distributed in a non-equidistant manner.
In order to allow for a structured observer design and a thorough stability analysis, the system is assumed to be restricted by the following characteristics that are common for similar approaches such as \cite{MazencMN22}.
\begin{ass}\label{as:sd_sys}
    System \eqref{eq:sd_sys} fulfills the conditions:
    \begin{enumerate}
        \item $\big(C,A(t)\big)$ is strongly observable such that $\exists P\in\cC^1(\mathbb{R}^+_0,\mathbb{R}^{n\times n})$ Lyapunov transformation w.r.t. the linear part as a regular backward-transform for the system from its OCF (see \eqref{eq:OCF} for linear case),
        \item The sampling periods are bounded such that\\
        \hspace*{0.7cm}$\exists\,\overline{T},\underline{T}>0\,\forall i\in\mathbb{N}:\underline{T}\leqslant t_{i+1}-t_i\leqslant\overline{T}$,
        \item The function $\phi(\cdot,u,t)$ is Lipschitz-continuous $\forall u(t)\in\mathbb{R}^m,t\geqslant0$, thus 
        $\exists L_\phi>0:$\\
        \hspace*{0.7cm}$\|\phi(y_1,u,t)-\phi(y_2,u,t)\|\leqslant L_\phi\|y_1-y_2\|~\forall y_{k}$\,.
\end{enumerate}
\end{ass}
\vspace{0.5ex}
For the algebraic observer design, nonlinearity $\phi$ is treated as $n$ known multiple inputs with the input matrix $B=I$, thus resulting in the following OCF parameterization and the corresponding differential operators: 
\begin{align}
    &a(t)\!=\!\mi P^{\mi1}(AP\!\mi\!\dot P)e_1,B_o(t)\!=\!P(t)\inv,E_o(t)\!=\!P(t)\inv E \nonumber\\[0.5ex]
     &\qquad\quad\Rightarrow L(t)[\,\cdot\,]=\big(\frac{\dd^{n}}{\dd t^{n}}[\,\cdot\,]+\cD (a(t)[\,\cdot\,])\big), \label{eq:sd_operators}\\[0.5ex]
     &\cB(t)[\,\cdot\,]=\cD \big(P(t)^{\mi1}[\,\cdot\,]\big),\,
     \cE(t)[\,\cdot\,]=\cD \big(P(t)^{\mi1}E[\,\cdot\,]\big) \nonumber
\end{align}
appearing in the input-output representation of system \eqref{eq:sd_sys} analogous to the argumentation related to the operator form \eqref{eq:inoutop}.\\
The central observer algorithm is composed of the metho\-dology presented in Theorem \ref{thm:mf_observer} in combination with the dynamic output prediction scheme explained in \cite{MazencMN20}.
Determining the convergence behavior quantitatively, the stability depends on two sources from different system components which are denoted by the following relevant coefficients:
\begin{itemize}
    \item Static model parameters (system input and output): \\ \hspace*{0.7cm} $K_1=\|C\|_2,K_2=\|CE\|_2$,
    \vspace{0.5ex}
    \item Varying model parameters (system dynamics): \\ \hspace*{0.7cm}
    $K_3=\|CA(\cdot)P(\cdot)\|_{{{\cLi(\mathbb{R}_0^+)}}}$
    \vspace{0.5ex}
    \item Local MF kernel gains (design) on moving horizon: \\ \hspace*{0.7cm}$\eta_a(t)=\|L^\ast(\cdot+t-T)[\vp]\|_{\cLi[0,T]},$\\
    \hspace*{0.7cm}$\eta_\phi(t)=\|\cB^\ast\vp\|_{\cLi[0,T]},$\,
    $\eta_d(t)=\|\cE^\ast\vp\|_{\cLi[0,T]}$.
    \vspace{0.5ex}
     \item Global MF kernel gains (design): \\ \hspace*{0.7cm}$\bar\eta_a=\|\eta_a\|_{\cLi(\mathbb{R}_0^+)},
    \bar\eta_\phi=\|\eta_\phi\|_{\cLi(\mathbb{R}_0^+)},$\\
    \hspace*{0.7cm}$\bar\eta_d=\|\eta_d\|_{\cLi(\mathbb{R}_0^+)}$.
\end{itemize}
The existence of such bounds on the compact moving horizon interval is guaranteed by the boundedness of $A(t)$ and the strong observability property from Assumption \ref{as:sd_sys}.
In the following theorem, convergence and estimation error bounds are provided in addition to the actual state estimator.
\vspace{0.5ex}
\begin{theorem}[Sampled-data state estimation]\label{thm:sd_observer}
    For a given UMF $\varphi:[0,T]\rightarrow\mathbb{R}^n$ of order $n$, the sampled-data system \eqref{eq:sd_sys} with constant initial function fulfilling Assumption \ref{as:sd_sys} and disturbances $d,\nu$ fulfilling Assumption \ref{as:disturb}, the state $x$ is reconstructed by
    \begin{equation}\label{eq:sdobs}
        \begin{aligned}
            \left\{\begin{array}{l}
                \dot{\hat{y}}=\,CA(t)\hat{x}+C\phi(\hat{y},u,t)\,\forall t\in[t_i,t_{i+1})\,,\\
                \hat{y}(t_i)=\,\tilde{y}(t_i)\,\forall t=t_i\,,\\
                \hat{x}=P(t)\big(-\lr{L^\ast\varphi}{\hat{y}}+\lr{\cB^\ast\varphi}{\phi(\hat{y},u,t)}\big)\,.
            \end{array}\right.
        \end{aligned}
    \end{equation}
    If the following inequality holds, that is
    \begin{equation}\label{eq:stab_crit}
\overline{T}\big(\underbrace{K_1L_\phi+K_3T(\bar\eta_a+\bar\eta_\phi L_\phi)}_{=:\lambda(\vp)}\big)<1
    \end{equation}
    where $T^\ast=\overline{T}+T$, then $\forall t\geqslant T^\ast$ the output prediction error $e_y(t)=y-\hat{y}$ is bounded by
    \begin{equation*}
        \|e_y(t)\|\leqslant\hspace{-1ex}\sup_{\tau\in[0,T^\ast]}\hspace{-0.5ex}\|e_y(\tau)\|\cdot\exp{\left[\frac{\ln(\overline{T}\lambda)}{T}(t\mi T^\ast)\right]}
        +\frac{W(t)}{1-\overline{T}\lambda}
    \end{equation*}
    with the disturbance impact
    \begin{equation*}
        W(t)=\|\nu\|_{\cLi[t\mi\overline{T}\!,t]}+\overline{T}\big(K_2+K_3T\,\bar\eta_d\big)\|d\|_{\cLi[t\mi T^\ast\!,t]}\,.
    \end{equation*}
    Furthermore, $\forall t\geqslant T^\ast+T$ the state estimation error $e_x(t)=x-\hat{x}$ satisfies 
    \begin{align}
        &\|e_x(t)\|\leqslant\left((\bar\eta_a\!+\!\bar\eta_\phi L_\phi)\!\left(\alpha_x\exp\!\left[{\frac{\ln(\overline{T}\lambda)}{T^\ast}(t\!\mi\! T^\ast\!\!\!\mi\! T)}\right]\right.\right.\nonumber\\
        &\hspace{2em}\left.\left.\hspace{0.2cm}+\frac{\|W\|_{\cLone[t-T,t]}}{1-\overline{T}\lambda}\right)+\bar\eta_d\|d\|_{\cLone[t-T,t]}\right)\|P\|_{\cLi} \label{eq:errorx}
    \end{align}
    with the positive constant
    \begin{equation*}
        \alpha_x=\frac{T^\ast}{\ln(\overline{T}\lambda)}\left((\overline{T}\lambda)^{\frac{T}{T^\ast}}-1\right)\sup_{\tau\in[0,T^\ast]}\|e_y(\tau)\|
    \end{equation*}
    guaranteeing an exponentially converging state estimate when \eqref{eq:stab_crit} is fulfilled.
\end{theorem}
Compared to Theorem \ref{thm:mf_observer}, the convergence rate is asymptotic as it is bounded by an exponential function due to the dynamic predictor approach. However, combining the observer \eqref{eq:mfobs} with the output prediction scheme is straightforward without adding further complexity for the implementation. The only challenge is the convoluted structure of the resulting closed-loop because of the cross-coupling between the algebraic moving horizon estimator and the output ODE. The related stability analysis is the main subject of this section.\\
The concept of the trajectory based stability criterion established in \cite{MazencM15} deals with nonlinear systems including delays and state history. Its alternative perspective is predicated on implicit inequalities with function norms over the considered finite time horizon rather than on energy-like derivatives. This enables the inclusion of algebraic approaches such as presented in Section \ref{sec:algobs}.\\
The starting point for the argumentation is similar to the ideal measurement case discussed for Theorem \ref{thm:mf_observer}. An exact representation of the state solution $x(t)$ is constructed via the modulation of the input-output behavior related to system \eqref{eq:sd_sys} with the nonlinearity $\phi$ treated as a known input.

\begin{proof}
Consider the state estimation error with respect to observer \eqref{eq:sdobs} and use the analog to the state representation from \eqref{eq:state_mod_form} to obtain
\begin{align}
    e_x=&\,x-\hat{x}=P(t)(z-\hat{z})\nonumber\\
    =&P(t)\big(-\lr{L^\ast\varphi}{y}+\lr{\cB^\ast\vp}{\phi(y,u,t)}+\lr{\cE^\ast\vp}{d} \nonumber\\
    &\quad+\lr{L^\ast\varphi}{\hat{y}}-\lr{\cB^\ast\vp}{\phi(\hat{y},u,t)}\big) \nonumber\\
    =&P(t)\big(\mi\lr{L^\ast\varphi}{e_y}+\lr{\cB^\ast\vp}{\phi(y,u,t)- \phi(\hat{y},u,t)}\nonumber\\
    &\hspace{2cm}+\lr{\cE^\ast\vp}{d}\big)\,. \label{eq:state_error}
\end{align}
The dynamics of the output error $e_y=y-\hat{y}$ then obey
\[
    \begin{aligned}
        \dot{e}_y=&\,CA(t)\overbrace{(x\!-\!\hat{x})}^{e_x}+C\overbrace{\left(\phi(y,u,t)-\phi(\hat{y},u,t)\right)}^{e_\phi}+CEd
        \\
        =&\,CA(t)P(t)\left(-\!\lr{L^\ast\vp}{e_y}+\lr{\cB^\ast\vp}{e_\phi}+\lr{\cE^\ast\vp}{d}\right)\\
        &\quad+C\,e_\phi+CE\,d\,,\\
        e_y(t_i)=&y-\tilde{y}=-\nu\,.
    \end{aligned}
\]
Subsequent integration of $\dot{e}_y(s)$ over $[t_i,t)$ results in
\begin{equation}\label{eq:ey_integr}
    \begin{aligned}
        \int_{t_i}^{t}\dot{e}_y\dd s={e}_y+\nu(t_i)=C\int_{t_i}^{t}\big[e_\phi(s)+Ed(s)+\\
        A(t)P(t)\,{\big(-\lr{L^\ast\varphi}{e_y}+\lr{\cB^\ast\vp}{e_\phi}+\lr{\cE^\ast\vp}{d}\big)}\big]\,\dd s
    \end{aligned}
\end{equation}
which contains modulated terms being integrated for a second time. In order to replace those during the central bound argumentation, an exemplary auxiliary calculation step of integrating the modulation operator with respect to the time-variant kernel $L^\ast(t)\vp$ applied to the signal $\xi\in\cLone^\text{loc}\cap\cL2^\text{loc}$ is performed while considering the multiplication with a bounded matrix-valued function $M:\mathbb{R}_0^+\rightarrow\mathbb{R}^{n\times n}$:
\begingroup
\allowdisplaybreaks
\begin{align}
    &\Big\|\int_{t_i}^{t} M(s) \lr{L^\ast\vp}{\xi} \dd s\Big\|\nonumber\\[-0.5ex]
    &\hspace{3ex}=\Big\|\int_{t_i}^{t}M(s)\!\int_{s\mi T}^{s}\!\! L^\ast(\tau)[\vp](\tau\!\mi\!s\pl T)\xi(\tau) \dd\tau\,\dd s\Big\| \nonumber\\
    &\hspace{3ex}\leqslant\int_{t_i}^{t}\left\|M(s)\int_{s\mi T}^{s} L^\ast(\tau)[\vp](\tau\!\mi\!s\pl T)\xi(\tau) \dd\tau\right\|\dd s \nonumber\\
    &\hspace{1ex}\overset{\text{H\"older}}{\leqslant}
    \sqrt{\int_{t_i}^{t}M(s)^\top M(s)\dd s} \nonumber\\
    &\hspace{1.2cm}\sqrt{\int_{t_i}^{t}\left(\int_{s\mi T}^{s} L^\ast(\tau)[\vp](\tau\!\mi\!s\pl T)\xi(\tau) \dd\tau\right)^2\!\dd s}\nonumber\\
    &\hspace{1ex}\overset{\text{Jensen}}{\leqslant}\sqrt{t-t_i}\,\|M\|_{\cLi}\nonumber\\
    &\hspace{1.2cm}\sqrt{\int_{t_i}^{t}T\int_{s\mi T}^{s} \|L^\ast(\tau)[\vp](\tau\!\mi\!s\pl T)\|^2\,\|\xi(\tau)\|^2 \dd\tau\dd s}\nonumber\\
    &\hspace{3ex}\overset{\text{H.}}{\leqslant}\sqrt{t-t_i}\,\|M\|_{\cLi}\sqrt{T}\nonumber\\
    &\hspace{1cm}\usqrt{\int_{t_i}^{t}\ubrace{\|L^\ast(\cdot\pl s\!\mi\!T)[\vp]\|^2_{\cLi[0,T]}}{\eta_a(s)^2}\int_{s\mi T}^{s} \|\xi(\tau)\|^2 \dd\tau\dd s}\nonumber\\
    &\hspace{3ex}\leqslant\sqrt{(t-t_i)T}\,\|M\|_{\cLi}\sqrt{\bar{\eta}_a^2\!\int_{t_i}^{t}T\,\|\xi\|_{\cLi[s\mi T,s]}^2\dd s}\nonumber\\
    &\hspace{3ex}\leqslant\sqrt{(t-t_i)}\,T\,\|M\|_{\cLi}\bar{\eta}_a\sqrt{(t-t_i)\|\xi\|_{\cLi[t_i\mi T,t]}^2}\nonumber\\
    &\hspace{3ex}\leqslant T\,\|M\|_{\cLi}\bar{\eta}_a(t-t_i)\|\xi\|_{\cLi[t_i\mi T,t]}.\nonumber
\end{align}
\endgroup
Using Assumption \ref{as:sd_sys} for bounding the reset interval, leads to the following norm consideration
\begin{equation}\label{eq:modint_bound}
    \Big\|\int_{t_i}^{t} M(s)\lr{L^\ast\vp}{\xi} \dd s\Big\|\leqslant
    \bar\eta_a\,{T}\,\|M\|_{\cLi}\,\overline{T} \|\xi\|_{\cLi[t\mi\overline{T}\mi T,t]}
\end{equation}
resulting in a fundamental $\cL2$ characterization which plays an essential role when bounding the error terms.
This procedure is performed for the other operator cases $\cB^\ast\vp$ and $\cE^\ast\vp$ analogously. The interval argument is omitted when considering the respective $\cLi$ norm in the following.

Combining these facts, the Euclidean norm $\|\cdot\|_2=\|\cdot\|$ of the output error $e_y$ from \eqref{eq:ey_integr} is bounded by
\begin{align*}
    &\|e_y\|\leqslant|\nu(t_i)|+\overbrace{\|C\|}^{=K_1}\int_{t_i}^{t}\|e_\phi\|\dd s+\overbrace{\|CE\|}^{=K_2}\int_{t_i}^{t}\|d\|\dd s \\
    &+\int_{t_i}^{t}\hspace{-1ex}\big\|CAP\big(\lr{L^\ast\varphi}{e_y}+\lr{\cB^\ast\vp}{e_\phi}+\lr{\cE^\ast\vp}{d}\big)\big\|\dd s\\
    &\overset{\text{\eqref{eq:modint_bound}}}{\leqslant}|\nu|+K_1\int_{t_i}^{t}\|e_\phi\|\dd s+K_2\int_{t_i}^{t}\|d\|\dd s\\
    &+T\overline{T}\underbrace{\|CAP\|_{\cLi}}_{=K_3}\Big(\bar\eta_a \|e_y\|_{\cLi[t\mi T^\ast,t]} +\bar\eta_\phi\|e_\phi\|_{\cLi[t\mi T^\ast,t]} \\[-3ex]
    &\hspace{3.5cm}+\bar\eta_d\|d\|_{\cLi[t\mi T^\ast,t]}\Big)\\
    &\overset{\text{As.\ref{as:sd_sys}}}{\leqslant}|\nu|+K_1L_\phi\int_{t\mi\overline{T}}^{t}\|e_y\|\dd s+K_2\int_{t\mi\overline{T}}^{t}\|d\|\dd s\\
    &+T\overline{T}K_3\Big(\bar\eta_a \|e_y\|_{\cLi[t\mi T^\ast,t]} +\bar\eta_\phi L_\phi\|e_y\|_{\cLi[t\mi T^\ast,t]} \\
    &\hspace{3.5cm}+\bar\eta_d\|d\|_{\cLi[t\mi T^\ast,t]}\Big)\\[-2.5ex]
    &\leqslant|\nu|+\overbrace{(t\mi t_i)}^{\leqslant \overline{T}}\hspace{-1ex}\sup_{s\in[t\mi\overline{T},t]}\hspace{-1ex}\big(K_1L_\phi\|e_y(s)\|+K_2\|d(s)\|\big)\\
    &+T\,\overline{T}K_3\!\!\sup_{s\in[t-T^\ast,t]}\hspace{-1ex}\big((\bar\eta_a+\bar\eta_\phi L_\phi)\|e_y(s)\|
    +\bar\eta_d\|d(s)\|\big)
\end{align*}
with maximum relevant horizon $T^\ast=\overline{T}+T$.
Increase the interval and take the supremum directly, thus
\[
    \sup_{\tau\in[t\mi\overline{T},t]}\hspace{-1ex}\|\xi(\tau)\|\,\leqslant\!\sup_{\tau\in[t\mi T^\ast,t]}\hspace{-1ex}\|\xi(\tau)\|
\]
which results in the following bound step:
\begin{align*}
    \|e_y\|\leqslant&\,\overline{T}\overbrace{\big(K_1L_\phi+K_3T(\bar\eta_a+\bar\eta_\phi L_\phi)\big)}^{\lambda(\vp)}\hspace{-1ex}\sup_{s\in[t-T^\ast\!\!,t]}\hspace{-1ex}\|e_y(s)\| \\
    &+\!\!\!\!\sup_{s\in[t-\overline{T},t]}\!\!|\nu(s)|+\overline{T}\big(K_2+K_3T\bar\eta_d\big)\hspace{-1ex}\sup_{s\in[t-T^\ast\!\!,t]}\hspace{-1ex}\|d(s)\|\,.
\end{align*}
As $\overline{T}\lambda<1$,the trajectory-based approach from \cite[Lem.1]{MazencM15} can be applied with its input-to-state stability (ISS) relation ending up to be $\forall t\geqslant T^\ast$:
\begin{align*}
    &\|e_y(t)\|\leqslant\sup_{\tau\in[0,T^\ast]}\|e_y(\tau)\|\cdot\exp{\left[\frac{\ln(\overline{T}\lambda)}{T^\ast}(t-T^\ast)\right]} \\
    &+\frac{1}{1\!\mi\!\overline{T}\lambda}\underbrace{\Big(\|\nu\|_{\cLi[t\mi\overline{T},t]}\!+\!\overline{T}\big(K_2\!+\!K_3T\,\bar\eta_d\big)\|d\|_{\cLi[t\mi T^\ast\!\!,t]}\Big)}_{W(t)}.
\end{align*}
To the end of state estimation, the bound on the observation error stated at the start of the proof  in \eqref{eq:state_error} may be expressed as the following applying a similar argumentation as in \eqref{eq:modint_bound}:
\begin{align*}
    \|e_x(t)\|\leqslant&\,\|P(t)\|\int_{t-T}^t\big(\eta_a(t)\|e_y(\tau)\|+\eta_\phi(t) L_\phi\|e_y(\tau)\|\\[-1ex]
    &\hspace{4.0cm}+\eta_d(t)\|d(\tau)\|\big)\dd\tau \\
    \leqslant&\, \|P(t)\|\int_{t-T}^t(\eta_a(t)+\eta_\phi(t) L_\phi)\cdot\\
    &\hspace{0.5cm}\cdot\left(\|e_y\|_{[0,T^\ast]}\,e^{\frac{\ln(\overline{T}\lambda)}{T^\ast}(\tau-T^\ast)} 
    +\frac{W(\tau)}{1-\overline{T}\lambda}\right)\dd \tau\\
    &\hspace{0.2cm}+\eta_d(t)\|P(t)\|\|d\|_{\cLone[t-T,t]} \\
    \overset{\text{H.}}{\leqslant}&\,\|P\|_{\cLi}\bigg[(\bar\eta_a+\bar\eta_\phi L_\phi)\left(\alpha_xe^{\frac{\ln(\overline{T}\lambda)}{T^\ast}(t\mi T^\ast\!\mi T)}\right.\\
    &\left.\hspace{1.4cm}+\frac{\|W\|_{\cLone[t-T,t]}}{1-\overline{T}\lambda}\right)+\bar\eta_d\|d\|_{\cLone[t-T,t]}\bigg]
\end{align*}
which holds for $t\geqslant T^\ast+T=\overline{T}+2T$ with coefficient originating from integration
\[
    0<\alpha_x=\frac{T^\ast\left((\overline{T}\lambda)^{\frac{T}{T^\ast}}-1\right)}{\ln(\overline{T}\lambda)}\sup_{\tau\in[0,T^\ast]}\|e_y(\tau)\|\,.
\]
Consequently, for the nominal case without disturbances $d,\nu\equiv0$, asymptotic convergence of the output prediction as well as the state estimator are guaranteed, hence
\[
    \begin{aligned}
        \|e_y\|\rightarrow0
        \Leftrightarrow \hat{y}\rightarrow y
        \quad\Rightarrow\quad {\|e_x\|\rightarrow0}
        \Leftrightarrow \hat{x}\rightarrow x\,.
    \end{aligned}
\]
This coincides with the statement in Theorem \ref{thm:sd_observer}.
\end{proof}
\vspace{1ex}
The convergence bounds provide valuable insight on the influence of the degrees of freedom represented by the MF tuning parameters. In the following remark, the consequences for different design choices are outlined.
\begin{remark}[Kernel design strategies]\label{rem:kernel_design}
    Besides its smoothness requirements, the UMF $\vp$ is mainly characterized by its boundary conditions which leaves the actually shape of the function and its resulting derivatives open for user selection.
    Consider the following cases with different specifications:
    \begin{itemize}
        \item The largest feasible sampling bound according to the stability goal is $\overline{T}<\frac{1}{\lambda}$ from \eqref{eq:stab_crit}. Thus, attenuating $\lambda$ is the main design goal of the sampled-data observer \eqref{eq:sdobs} for achieving higher robustness against the highest possible sampling bound $\overline{T}$.
        \item Reducing gains $\bar\eta_a=\|L^\ast\vp\|_{\cLi}$ and $\bar\eta_\phi=\|\cB^\ast\vp\|_{\cLi}$ also increases the stability margin $|1-\overline{T}\lambda|$ by minimizing the cost function $J(\vp)=\bar\eta_a+L_\phi\bar\eta_\phi$ over $\vp$ ultimately decreasing $\lambda(\vp)$ in the process.
        \item The measurement noise $\nu$ shows a direct impact on $e_y$ through the factor $\frac{1}{1-\overline{T}\lambda}$. Thus, a minimization of $J$ also leads to smaller amplification of $\nu$.
        Additionally, a direct multiplication of $J$ can be seen within the state reconstruction error bound $\|e_x\|$, further diminishing the impact on the estimate for smaller choices and allowing some filtering.
        \item On the other hand, mitigating the disturbance impact of $d$ can be achieved by reducing $\bar\eta_d=\|\cE^\ast\vp\|_{\cLi}$.
        \item The horizon length $T>0$ plays an important role in all the above considerations, especially as it directly determines the convergence rate of the algebraic observer part while directly appearing inside of $\lambda$. It is recommended to use it as a simple tool for intuitive final tuning.
    \end{itemize}
    With a given system parameterization $\{A(t),L_\phi,C,E\}$ of \eqref{eq:sd_sys}, these sensitivity measures can be optimized with respect to the kernel function $\vp$ which can be done in closed form for the LTI case.
\end{remark}
\vspace{0.5ex}
The subsequent numerical study illustrates the observer behavior and the impact of system and disturbance characterizations on the stability margins. 

\section{Implementation}
\label{sec:implementation}
In order to summarize the design and implementation procedure for the MF-based estimator for the sampled-data system described in Theorem \ref{thm:sd_observer}, this section discusses some of the practical aspects with regards to the numerical realization of the approach.
To demonstrate its functionality, the methodology is applied to an example problem and compared to a similar scheme.

\subsection{MF-based Observer Realization}
Considering the standard MF observer from \eqref{eq:mfobs}, it shall be noted that the major part of the necessary calculations is done offline, while during runtime only the integration with respect to the predefined kernels has to be implemented. This operation \eqref{eq:modop} can efficiently be realized via an FIR filter structure by discretization via numeric integration or, alternatively, in the continuous-time fashion established in \cite[Sec.4.2]{IonesiRJ19}.
Similarly, the extended version \eqref{eq:sdobs} shows the same structure with the addition of the low-order dynamic output predictor, which does not add much computational cost.\\
For a given system \eqref{eq:sd_sys} with known bounds as from Assumption \ref{as:sd_sys}, the preparatory design process contains the following calculation steps:
\vspace{-1ex}
\begin{enumerate}
    \item ONF: $P\Rightarrow a,B_o=P^{-1},E_o=P^{-1}E$,
    \item Determine UMF $\vp$ of order $n$ with BCs \eqref{eq:umfcond},
    \item \textit{Optional:} Optimizing kernel $\vp$ via $\min(\bar\eta_a+\bar\eta_\phi L_\phi)$ for stability margin or $\min\eta_d$ for robustness, 
    \item Determine data gains $L^\ast\vp(\tau)$ and $\cB^\ast\vp(\tau)$ (e.g. FIR),
    \item Check stability via condition of Theorem \ref{thm:sd_observer}.
\end{enumerate}
\vspace{-1ex}
All of these steps are performed offline. An example of a resulting kernel $\vp$ of order $n=2$ is presented in Figure~\ref{fig:umf}.
\begin{figure}[!t]
    \centering
    \includegraphics[width=0.99\columnwidth]{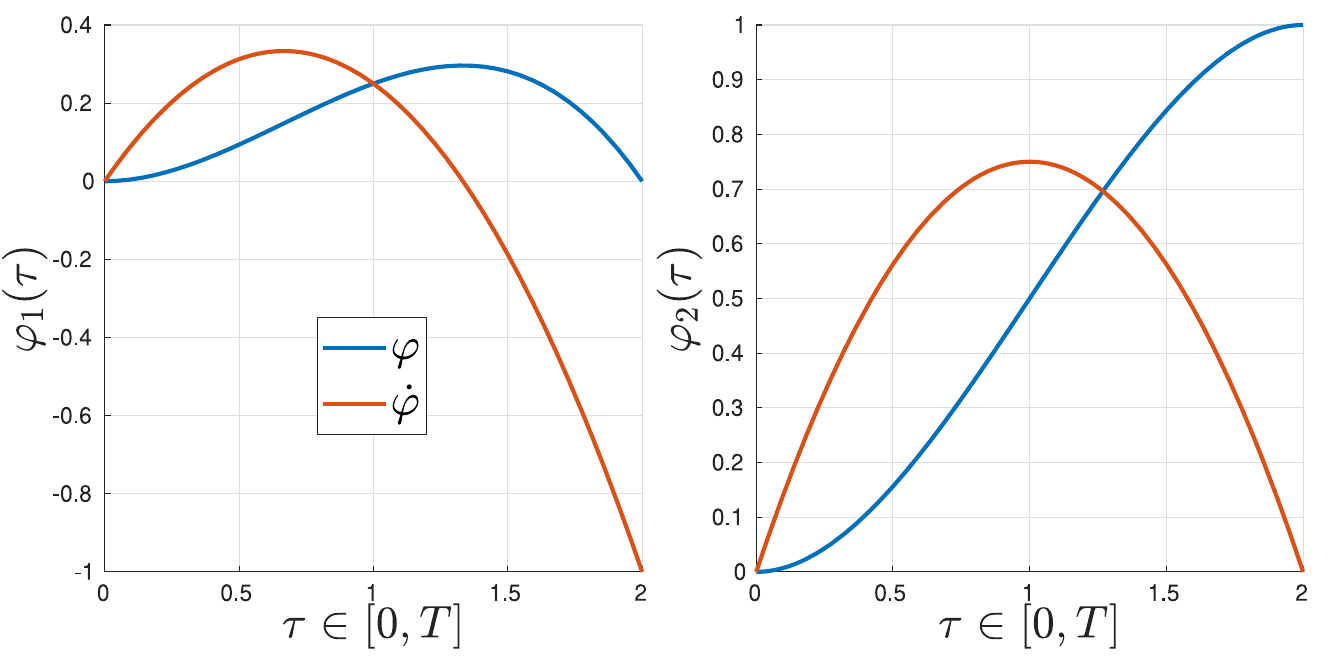}\vspace{-0.25cm}
    \caption{UMF kernel $\varphi(\tau)=[\varphi_1,\varphi_2]^\top$ of order 2 and its derivative $\dot{\vp}$ as defined in Equation \eqref{eq:umfcond} of Definition \ref{def:mf}.}
    \label{fig:umf}
\end{figure}
Note that the characteristic UMF BCs are fulfilled with $\vp(T)=e_2,\dot\vp(T)=-e_1$ guaranteeing the basic functionality of the estimator.
For illustrating the realization of observer \eqref{eq:sdobs}, Figure \ref{fig:block_diagr_impl} depicts its block diagram architecture.
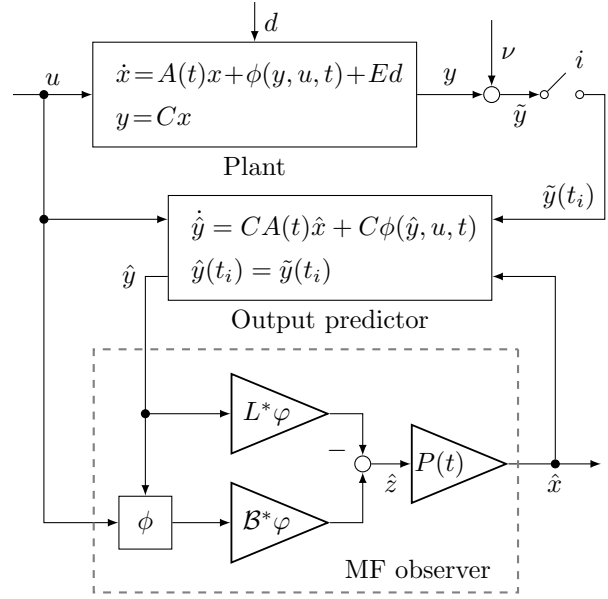
\begin{figure}[ht]
    \centering
    \tikzstyle{block} = [draw, fill=white, rectangle, minimum height=3em, minimum width=3em, anchor=center]
\tikzstyle{sum} = [draw, fill=white, circle, minimum height=0.6em, minimum width=0.6em, anchor=center, inner sep=0pt]
\tikzstyle{triang} = [draw, isosceles triangle, inner sep=1pt, line width = 0.7]

\begin{center}
    \begin{tikzpicture}[auto, scale=0.95, >=latex]
	
	\node at (0,0) (IN){};
	\node[block, right = 3em of IN] (SYS){~$\begin{array}{l}
             \dot{x}\!=\!A(t)x\!+\!\phi(y,u,t)\!+\!Ed \\
             y\!=\!Cx
        \end{array}$~};
	\node [below] at (SYS.south) {Plant};
        \node[sum, right = 2.5em of SYS] (SUMNU){};

        \node [above=1.5em of SYS] (DIN){};
        \node [above=2.5em of SUMNU] (NUIN){};
        \draw[->] (IN) -- node[above]{$u$} (SYS);
        \draw[->] (SYS) -- node[above]{$y$} (SUMNU);
        \draw[->] (DIN) -- node[right]{$d$} (SYS);
        \draw[->] (NUIN) -- node[right]{$\nu$} (SUMNU);

        \node[sum, right = 1.5em of SUMNU, minimum height=0.3em, minimum width=0.3em] (SWITCH1){};
        \node[sum, right = 1em of SWITCH1, minimum height=0.3em, minimum width=0.3em] (SWITCH2){};
        \node [above=0.5em of SWITCH2] (SWITCH3){$i$};
        \draw[->] (SUMNU) -- node[below]{$\tilde{y}$} (SWITCH1);
        \draw[-] (SWITCH1) -- (SWITCH3);
        
        \node[block] at ($(SYS.south)+(3em,-4em)$) (OUTPREDICTOR){~$\begin{array}{l}
             \dot{\hat{y}}=CA(t)\hat{x}+C\phi(\hat{y},u,t) \\
             \hat{y}(t_i)=\tilde{y}(t_i)
        \end{array}$~};
        \node [below] at (OUTPREDICTOR.south) {Output predictor};
        \node[sum, fill=black, minimum size=0.3em, inner sep=0pt, right = 1em of IN] (SPLITU){};
        \draw[->] (SWITCH2) -- +(1em,0) |- node[above left]{$\tilde{y}(t_i)$} (OUTPREDICTOR.10);
        \draw[->] (SPLITU) |- (OUTPREDICTOR.170);

        \node[triang, anchor=east] at ($(OUTPREDICTOR.south)+(0em,-4.5em)$) (LPHI){~$L^\ast\vp$~};
        \node[triang, below = 2em of LPHI] (PHIB){~$\cB^\ast\vp$~};
        \node[block, left = 2.2em of PHIB, minimum height=2em, minimum width=2em] (PHIFUN){$\phi$};
        \node[sum] at ($(LPHI)+(4em,-2em)$) (SUMMF){};
        \node[triang, right = 1.5em of SUMMF] (PGAIN){$P(t)$};

        \node[sum, fill=black, minimum size=0.3em, inner sep=0pt] at (OUTPREDICTOR.170-|SPLITU) (SPLITU2){};
        \node[sum, fill=black, minimum size=0.3em, inner sep=0pt] at (LPHI-|PHIFUN) (SPLITY){};
        \draw[->] (OUTPREDICTOR.190) -| node[left]{$\hat{y}$} (PHIFUN);
        \node[sum, fill=black, minimum size=0.3em, inner sep=0pt, right = 1.6em of PGAIN] (SPLITX){};
        
        \draw[->] (SPLITX) |-  (OUTPREDICTOR.350);
        \draw[->] (SUMMF) -- node[below]{$\hat{z}$} (PGAIN);
        \node[above left = -0.4em and 0.0em of SUMMF](){$-$};
        
        \draw[->] (LPHI) -| (SUMMF);
        \draw[->] (PHIB) -| (SUMMF);
        \draw[->] (PHIFUN) -- (PHIB);
        \draw[->] (SPLITU2) |- (PHIFUN);
        \draw[->] (SPLITY) |- (LPHI);

        \node[right = 3.5em of PGAIN] (XOUT){};
        \draw[->] (PGAIN) -- node[below]{$\hat{x}$} (XOUT);

        \draw[gray,thick,dashed] ($(PGAIN.north west)+(1.5,1.2)$)  rectangle ($(PHIFUN.south east)+(-1.1,-0.6)$);
        \node [below] at ($(PGAIN.south west)+(0.1,-0.8)$) {MF observer};
        
    \end{tikzpicture}
\end{center}
    \caption{Block diagram of implementation scheme for observer \eqref{eq:sdobs}.}
    \label{fig:block_diagr_impl}
\end{figure}
Since the modulation operator FIR filter structures are numerically realized with a memory block and a pre-calculated gain, they are represented as triangular gains emphasizing the simplicity and real-time computability of the MF method.\\

\subsection{LTI Simulation Example with Input-Output Injection and Measurement Noise}
The developed algorithm is assessed in the same context as the example system discussed in \cite{MazencMN20}, i.e. a comparison to a related approach is made. An LTI system is considered for demonstrating how many of the presented calculation steps developed the LTV case can be simplified substantially.
Consider the plant model
\[
    \dot x
    =\underset{A}{\underbrace{\begin{pmatrix}0&1\\ -1&0\end{pmatrix}}}x+\underset{\phi(y,u)}{\underbrace{\begin{pmatrix}\epsilon\sin(y)+u^2\\-u\end{pmatrix}}}\,,
    y=x_1\,.
\]
 with $\epsilon>0$. The system is already in ONF, thus $P=I$. For the following study, the system parameters are chosen as in \cite{MazencMN20}, namely $L_\phi=\epsilon=\frac16$ and $\overline{T}=0.22$ with respect to a constant sampling rate $t_{i+1}-t_i=\overline{T}$. The latter is selected due to the specific observer and its stability condition, which ends up considering an equivalent horizon length of $\tau=\frac\pi2$. For the MF-based algorithm, $T=2$ is chosen with the corresponding kernel shape shown in Figure \ref{fig:umf}.
 Hence, both approaches are tuned to just marginally fulfill the closed-loop stability condition $\overline{T}\lambda<1\Leftrightarrow\lambda\approx4.5$.
The demonstration and comparison are performed for the noisy case with $\nu\in\cN(0,10^{-2})$ corresponding to an SNR value of $29.6\,\text{dB}$. The different output signals, which are labeled in the block diagram (Figure \ref{fig:block_diagr_impl}), are presented in Figure \ref{fig:output_noise}.
\begin{figure}[!t]
    \centering
    \includegraphics[width=0.95\columnwidth]{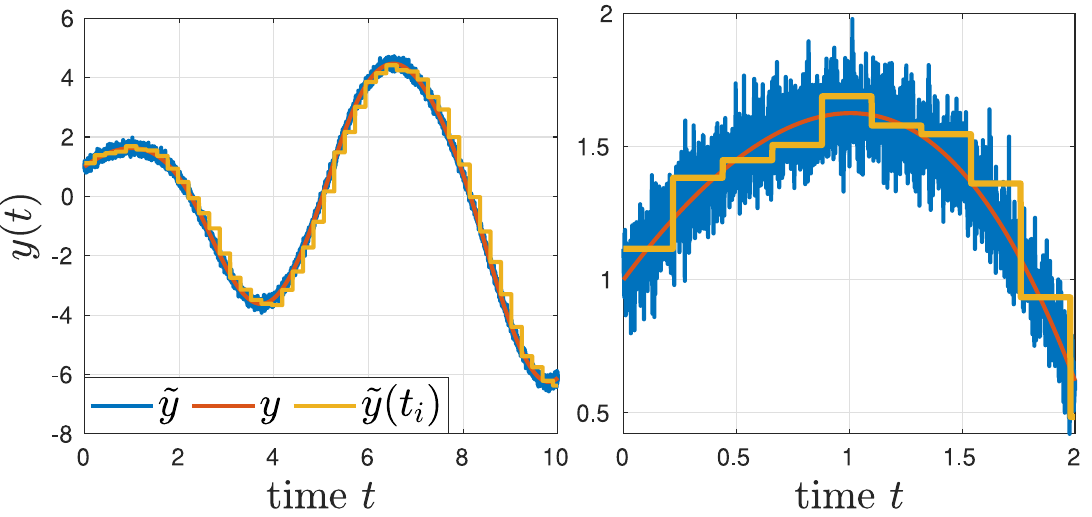}\vspace{-0.25cm}
    \caption{Output signal $y$, noisy measurement $\tilde{y}$ and sampled data $\tilde{y}(t_i)$ (time axis zoom right).}
    \label{fig:output_noise}
\end{figure}
The challenge of dealing with the perturbed sensor data that is only available at the equidistant sampling points becomes evident. Applying the dynamic output predictor improves the quality of the output data compared to only using the sampled data as demonstrated in Figure \ref{fig:ey_cmp}.
\begin{figure}[!t]
    \centering
    \includegraphics[width=1\columnwidth]{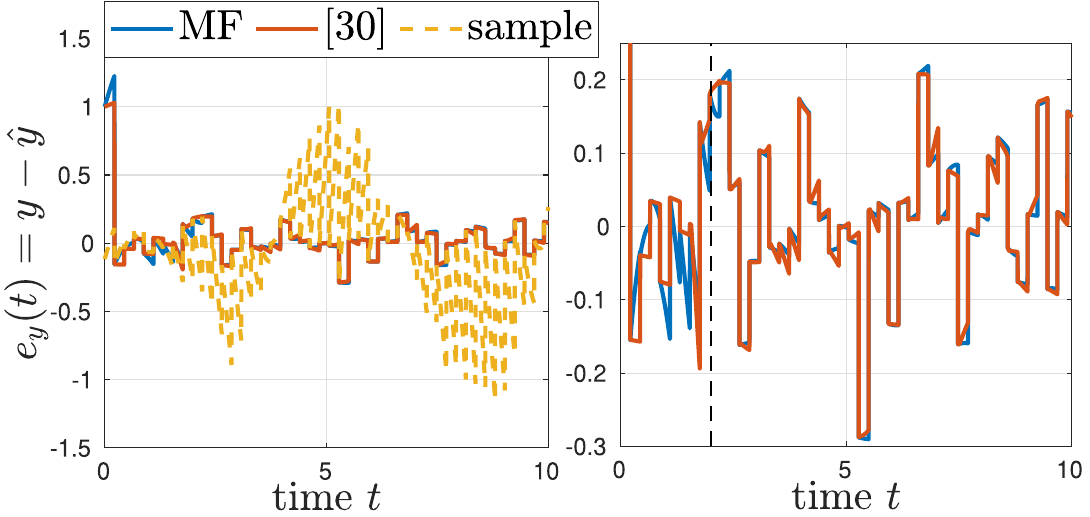}\vspace{-0.25cm}
    \caption{Comparison of different output errors $e_y=y-\hat{y}$ for the scenarios of observer \eqref{eq:sdobs}, the approach \cite{MazencMN20} and for just considering sampled data $\tilde{y}(t_i)$ in the noisy case (y-axis zoom right).}
    \label{fig:ey_cmp}
\end{figure}
It is notable that both algebraic approaches lead to nearly a congruent output behavior. Based on this prediction and past information, the observer part reconstructs the current state value which is verified in Figure \ref{fig:state_estim_noise}.
\begin{figure}[!t]
    \centering
    \includegraphics[width=0.95\columnwidth]{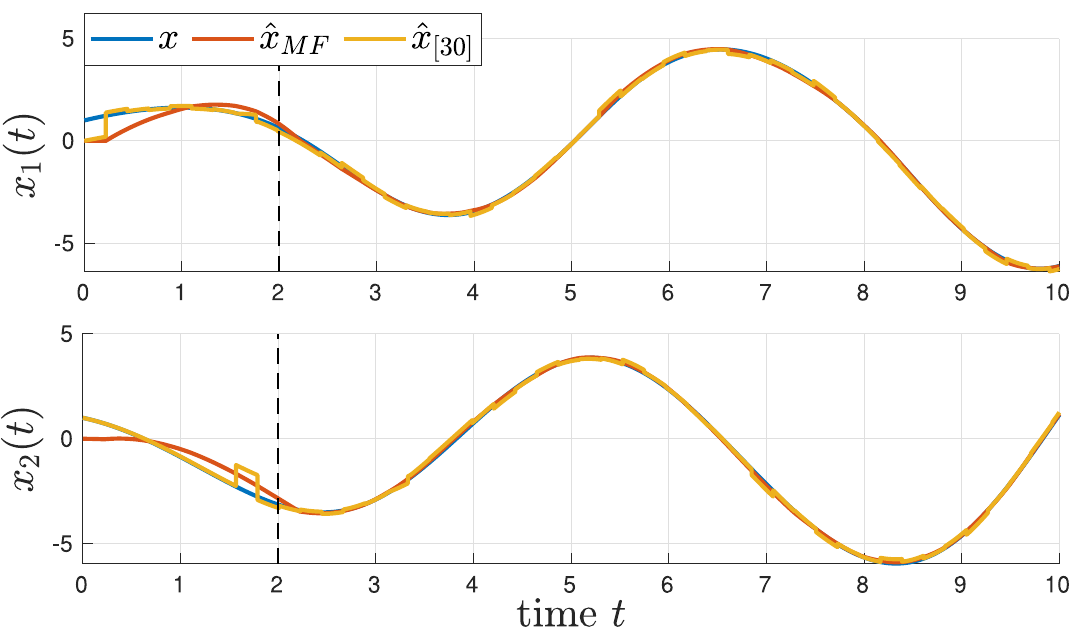}\vspace{-0.25cm}
    \caption{State estimation result $\hat{x}(t)$ using different algorithms.}
    \label{fig:state_estim_noise}
\end{figure}
After filling the data horizon with $T=2$ marked by a dashed line, both approaches capture the unknown state similarly. Figure \ref{fig:ex_cmp} gives more insight into the respective estimation error norms.
\begin{figure}[!t]
    \centering
    \includegraphics[width=1\columnwidth]{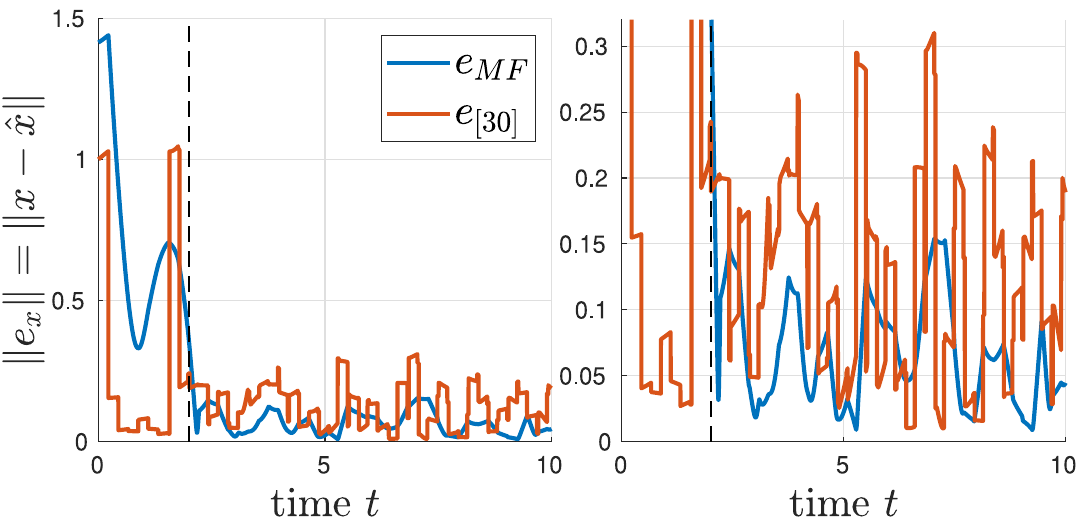}\vspace{-0.25cm}
    \caption{Norm of state estimation error $e_x=x-\hat{x}$ for the observer \eqref{eq:sdobs} and the approach \cite{MazencMN20} in the noisy case (y-axis zoom right).}
    \label{fig:ex_cmp}
\end{figure}
The impact of the noise $\nu$ is evident. Without it, the typical asymptotic convergence behavior can be observed. The MF-based observer is able of mitigating its effect in a better way due to the various tuning options it provides. While the estimator in \cite{MazencMN20} depends only on the shift parameter $\tau$, the approach \eqref{eq:sdobs} exhibits the corresponding parameter $T$ and besides that, offers some filtering options through the variation of the gains $\eta_a,\eta_\phi$ and $\eta_d$ discussed in Remark \ref{rem:kernel_design}. These can directly be computed and optimized over the UMF $\varphi$ due to the time-invariant characteristic input and output coefficients.\\
In order to assess the conservativeness of the stability bound $\overline{T}\lambda(\vp,\overline{T})$, the sampling delay can be increased to a point where the close-loop setup starts becoming unstable. For the given example, this happens around a delay value of $10\overline{T}=2.2$ indicating that there remains space for a more accurate bound.
However, from the experience of the simulation studies that were carried out, the criterion turns out to be practically relevant.\\

\subsection{Time-variant Pendulum Dynamics with Horizontal Disturbance}
\label{ssec:tvex}
The second simulation showcase demonstrates the presented framework for a time-variant system case and includes the evaluation of the calculated stability bounds from Theorem \ref{thm:sd_observer} discussing their conservativeness. Inspired by \cite[Sec.5.1]{MazencMN22}, the following equation of motion for a pendulum of length $L\!=\!1$ and mass $M\!=\!1$ whose suspension point is subjected to a time-varying, bounded, horizontal acceleration $h(t)$  with bounded time-variant friction $0\leqslant k(t)=\frac{1+\sin(t)}{10}\leqslant\bar{k}=0.2$ \cite[p.627,Exercise14.5]{khalil2002nonlinear} is considered:
\begin{equation*}
ML\ddot{\theta}+Mg\sin(\theta)+k(t)L\dot\theta=\frac{T_c(t)}{L}+Mh(t)\cos(\theta)
\end{equation*}
where $g$ is the gravitational constant. The horizontal force is unknown and treated as perturbation.
The corresponding state space form \eqref{eq:sd_sys} with $x=[\theta,\dot\theta]^\top,u=T_c,d=\frac{h(t)}{L}\cos{x_1}$ is obtained as:
\begin{equation*}
    \begin{aligned}
        &\dot x =\overbrace{\begin{pmatrix}0&1\\0&\mi\frac{k(t)}{M}\end{pmatrix}}^{=A(t)}x
        +\overbrace{\begin{pmatrix}0\\\mi\frac{g}{L}\sin\left(\frac{y}{c_o}\right)+\frac{u}{ML^2}\end{pmatrix}}^{=\phi(y,u)}+\overbrace{\begin{pmatrix}0\\1\end{pmatrix}}^{=E}d \\
        &y(t_i) = \,c_o\theta(t_i)+\nu
        =\underbrace{\begin{pmatrix}c_o&0\end{pmatrix}}_{=C}x(t_i)+\nu
    \end{aligned}
\end{equation*}
with the known output coefficient $c_o=2>0$.
For applying the developed methodology, the transformation of the linear part into OCF is pursued as $(A,C)$ here is nearly in canonical observability form. Checking the observability matrix referred to in Assumption \ref{as:obsvb}:
\begin{equation*}
    \cO = \begin{pmatrix}C\\\dot{C}+CA\end{pmatrix} = \begin{pmatrix}c_o&0\\0&c_o\end{pmatrix}
    ~\Rightarrow~ \det(\cO)=c_o^2=\delta
\end{equation*}
the required strong observability becomes clear. Thus, the parameterization of the OCF \eqref{eq:OCF} is calculated in the following way:
\begin{equation*}
    \begin{aligned}
         &q=\cO\inv e_2=\begin{pmatrix}0\\\frac{1}{c_o}\end{pmatrix} \\[-1ex]
        &\hspace{2em}\Rightarrow P = \begin{pmatrix}q&&A(t)q-\dot q\end{pmatrix}J_2=
            \left(\begin{array}{cc}
            \frac{1}{c_o } & 0\\
            -\frac{k\left(t\right)}{M\,c_o } & \frac{1}{c_o }
            \end{array}\right)\\
        &\hspace{2em}\Rightarrow A_o=P\inv(AP-\dot{P})=\left(\begin{array}{cc}
            -\frac{k\left(t\right)}{M} & 1\\
            \frac{\dot k\left(t\right)}{M} & 0
            \end{array}\right),\\[-0.5ex]
        &\hspace{1em}B_o=P\inv I=\left(\begin{array}{cc}
            c_o  & 0\\
            \frac{c_o \,k\left(t\right)}{M} & c_o 
            \end{array}\right),\,
        E_o=P\inv E=\begin{pmatrix}0\\c_o\end{pmatrix} 
    \end{aligned}
\end{equation*}
and accordingly, $C_o=\begin{pmatrix}1&0\end{pmatrix}$ and $a(t)=\left[\frac{k\left(t\right)}{M},\mi\frac{\dot k\left(t\right)}{M}\right]^\top$. These characteristic coefficients relate to the right PDO rearrangement with the input-output dynamic relation \eqref{eq:inout}:
\begin{equation*}
    \ddot{y}+\left(\frac{k(t)}{M}y\right)^{(1)}\!\!\!-\left(\frac{\dot k(t)}{M}y\right)
    \!=\!\left(B_1^\top(t)\phi\right)^{(1)}\!\!+\left(B_0^\top(t)\phi\right)+c_od
\end{equation*}
from which the operator form \eqref{eq:inoutop} is derived with the resulting expressions for $L,\cB,\cE$ provided in \eqref{eq:sd_operators}.
The characteristic bounds used in Theorem \ref{thm:sd_observer} are collected here:
\begin{itemize}
    \item $K_1=\|C\|=c_o$,
    \item $K_2=\|CE\|=0$,
    \item $K_3=\|CAP\|_{\cLi}=\underset{t\in\mathbb{R}_0^+}{\sup}\sqrt{1+\frac{k(t)^2}{M^2}}=\sqrt{1+\frac{\bar k^2}{M^2}}$, \\[-1ex]
    \item $L_\phi=\max\left\|\frac{\partial \phi}{\partial y}\right\|=\frac{g}{c_oL}$.
    \item $T=1$ for comparability to horizon in \cite{MazencMN22},
    \item $t_{i+1}-t_i=\overline{T}=0.02$ for simulation purposes (see Assumption \ref{as:sd_sys}).
\end{itemize}
The MF dependent gains $\bar\eta_a,\bar\eta_\phi,\bar\eta_d$ originate from the selected kernel fulfilling the UMF conditions from Definition \ref{def:mf}. As outlined in Remark \ref{rem:kernel_design}, an optimal selection can be performed by minimizing the cost function $J(\vp)=\bar\eta_a+L_\phi\bar\eta_\phi$ over $\varphi\in\cC^n([0,T]\in\mathbb{R}^{n\times n})$ subject to the UMF boundary conditions. A sum of norms cost functional like that, is not trivial to optimize and a finite dimensional subspace could be selected. Such a step leads to the MF presented in Figure \ref{fig:umf_opt} based on the given parameter values for the system.
\begin{figure}[!t]
    \centering
    \includegraphics[width=0.99\columnwidth]{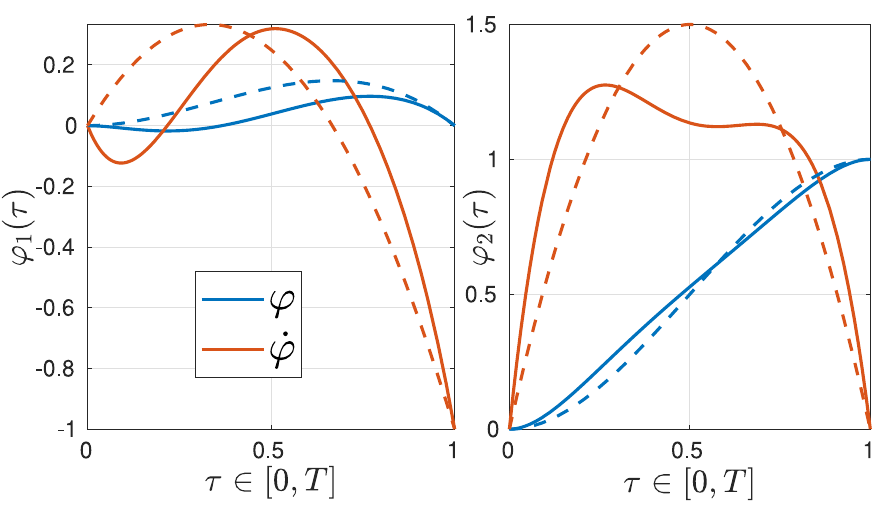}\vspace{-0.25cm}
    \caption{Sub-optimal UMF kernel $\varphi$ w.r.t. stability margin \eqref{eq:stab_crit} containing the cost term $J=\bar\eta_a+L_\phi\bar\eta_\phi\approx13.8$ (original kernels from Figure \ref{fig:umf} in dashed lines with $J\approx16.1$).}
    \label{fig:umf_opt}
\end{figure}
An improvement in comparison to the former MF shown in Figure \ref{fig:umf} can be noted. However, the maximum norm formulation does not leave arbitrary space for enhancing the norms.
In order to get an impression of the MF gain behavior, Figure \ref{fig:mf_bound_plot} illustrates the different time-varying norms.
\begin{figure}[!t]
    \centering
    \includegraphics[width=0.99\columnwidth]{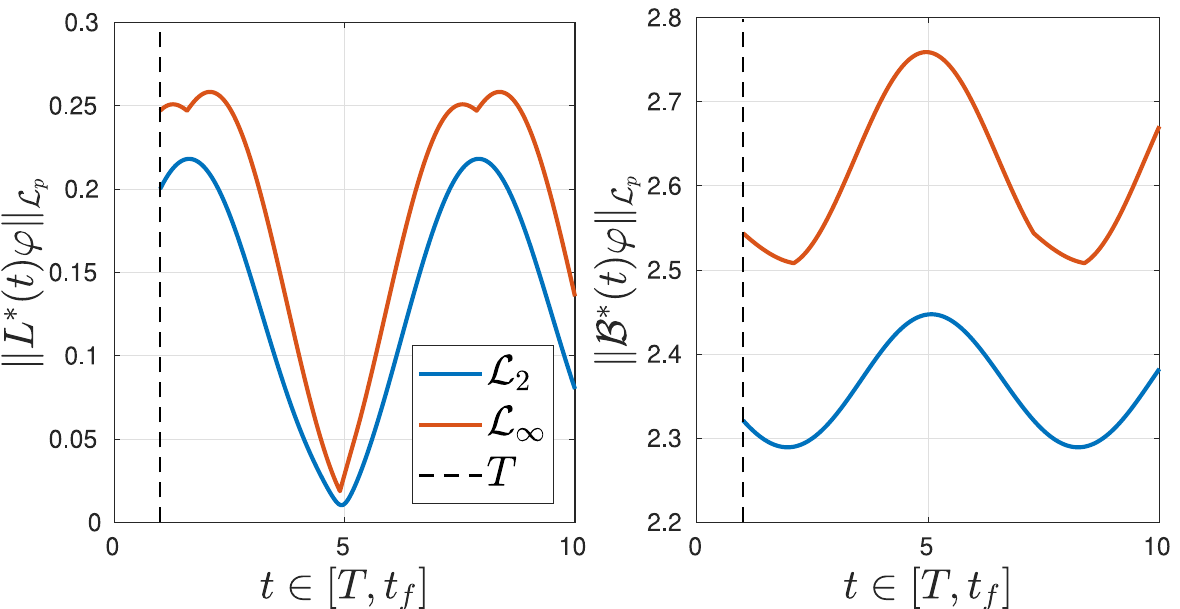}\vspace{-0.25cm}
    \caption{Time-variant function norms in $\mathcal{L}_p[0,T],p\in\{2,\infty\}$ of the adjoint operation specified in Lemma \ref{lem:shift} applied to the MF with the characteristic gains from the error considerations in Theorem \ref{thm:sd_observer} resulting in $\bar\eta_a\approx0.26$ and $\bar\eta_\phi\approx2.76$.}
    \label{fig:mf_bound_plot}
\end{figure}
Due to the boundedness of $A(t)$ and the strong observability, the MF gains show a periodic behavior with the worst case being taken into account for the stability criterion \eqref{eq:stab_crit}.
Focusing on the right graph in Figure \ref{fig:mf_bound_plot} related to the input operator $\cB$, it can be noted that the norms are relatively large in magnitude compared to the left side. This is partially due to the lack of structure in regards to the nonlinear term $\phi(y,u,t)$ which is considered as a vector input with input matrix $B=I$ resulting in components of $B_o$ that are not necessary for representing the PDO. Thus, the norm bounding procedure has potential for reduced and more accurate margins by imposing additional structure on the nonlinearity input similar to the impact of disturbance $d$. There, incorporating the knowledge about $E$ leads to $K_2=0$ instead of equalling $c_o$ in the alternative unstructured case $E=I$.
For evaluating the largest admissible sampling bound $\overline{T}<\frac{1}{\lambda}$, consider the following:
\begin{equation*}
    \begin{aligned}
        \lambda=&K_1L_\phi+K_3T(\bar\eta_a+\bar\eta_\phi L_\phi)\\
        =&\frac{g}{L}+T\sqrt{1+\frac{\bar k^2}{M^2}}\left(\bar\eta_a+\frac{\bar\eta_\phi g}{c_oL}\right)
        \approx 23.9 \\
        \Rightarrow \quad&\overline{T}\approx0.042\,.
    \end{aligned}
\end{equation*}
In order to set this result in relation to the discussion presented in \cite{MazencMN22}, consider the presented criterion with an adaptation of the notation and a minor correction including the Lipschitz constant (check \cite[obtaining Equ.(48) from (47)]{MazencMN22}):
\begin{equation*}
    \begin{aligned}
        \lambda_{\!\text{\cite{MazencMN22}}}=&\|CA\Psi\|_\infty\left(\sqrt{n}+\sqrt{\tau}{\color{black}L_\phi}\sqrt{\sum_{j=1}^{n-1}j\sigma_j}\right)+K_1L_\phi\\
        \approx&\,19.1
        \quad\Rightarrow\quad\overline{T}\approx0.052
    \end{aligned}
\end{equation*}
with $\sigma_j=\sup\{\|C\Phi_A(r,s)\|:0\leqslant r\leqslant s\leqslant r+j\tau\}$ related to the transition matrix $\Phi_A(t_1,t_0)$ and $\Psi$ representing the inverse of the shift-based observability matrix with respect to the shift parameter $\tau$.
In both cases, the maximum bounds on the sampling time $\overline{T}$ for which convergence can be guaranteed using by means of the trajectory-based approach result in values of similar scope. Nevertheless, the deviation between those stability margins is due to the distinct exact observer algorithms and related error bounding strategies.\\
For evaluating the accuracy or respective conservativeness of the stability criterion, the state estimation error bounded by the right-hand side of the inequality \eqref{eq:errorx} (RHS) is validated. To this end, different disturbance levels of $d(t)=\frac{h(t)}{L}\cos{(x_1(t))}$ are considered by taking a set of constant vertical accelerations $h\equiv H=\text{const.}$ into account. Note that $d(t)$ remains a time-varying perturbation term.
Figure \eqref{fig:bnd_compare} illustrates the resulting state estimation error norm $\|e_x\|=\|x-\hat{x}\|$ obtained by applying the observer algorithm \eqref{eq:sdobs} from Theorem \ref{thm:sd_observer} to the pendulum system with simulation parameters from Table \ref{tab:param} and puts it into the context of the corresponding RHS bound for a selection of disturbance values.
\begin{figure}[!t]
    \centering
    \includegraphics[width=0.99\columnwidth]{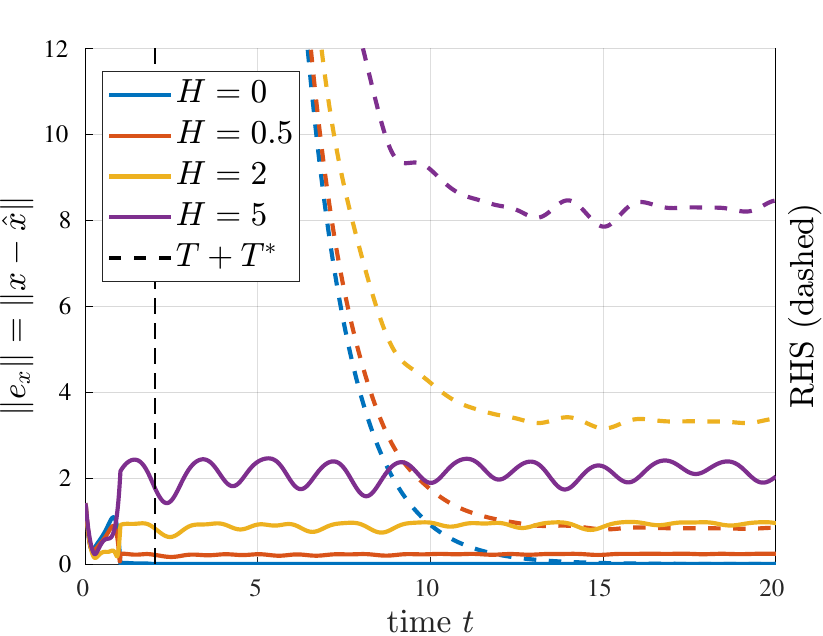}\vspace{-0.25cm}
    \caption{State estimation error norm $\|e_x(t)\|$ in comparison to right-hand side of bound relation \eqref{eq:errorx} (RHS) in dashed lines for different disturbance levels $H\in\{0,0.5,2,5\}$.}
    \label{fig:bnd_compare}
\end{figure}
In the absence of perturbations (blue line), the estimation error and the related dashed bound asymptotically converge to zero as proved for the given sampling delay $\overline{T}$. Notably, the error as well as its bound increase for larger disturbance amplitudes. The accuracy of the bound represented by the deviation from the real error trajectory decreases with rising $H$ values and thus, becomes increasingly conservative.
The RHS starting values at $t=T+T^\ast$ not displayed in Figure \eqref{fig:bnd_compare} lie in a similar numeric range of $300$ for the given system configuration.
Considering the stability result related to the criterion \eqref{eq:stab_crit}, it is worth pointing out that the bound characterizations like \eqref{eq:state_error} become more conservative as $\overline{T}\lambda\rightarrow1$ leading to decreasing coefficients within the exponential argument which consists of $\ln(\overline{T}\lambda)(t-T^\ast)/T^\ast$. Accordingly, the asymptotic bound convergence rate starts diminishing towards zero. Overall, it can be noted that the quantitative error assessment provided by Theorem \ref{thm:sd_observer} forms a helpful tool for convergence and robustness analysis while loosing precision when confronted with circumstances such as relatively large disturbances or marginal stability leading to conservative characterizations. Nevertheless, it remains a crucial indicator for design purposes and robustly selecting tuning parameters.\\
For a concluding comparison between the two non-asymptotic observer approaches from \cite{MazencMN22} and the methodology presented here, several similarities but also distinct features can be noted.
Due to their algebraic structure, both observers ensure an exact state representation after a fixed time $T=(n\!\mi\!1)\tau$. The finite time horizon serves as a simple tuning parameter for balancing convergence time against effects such as sensor noise. In both works, the fundamental assumptions for guaranteeing a stable sampled-data configuration turn out to be equivalent.
However, the different formulations of the time-variant observability condition are related to different perspectives on the problem representation. While in \cite{MazencMN22}, a shift-based criterion relying on the computation of the transition matrix is applied, the more conventional derivative-based transformation path is used here. Effectively, this comes down to the design decision on whether to integrate or differentiate the involved vector fields from the dynamic model in form of the given matrix functions (e.g. integrability or smoothness).
A notable feature of the MF framework are the kernel shapes and the related gains which still leave space for additional tuning and systematic optimization.
Regarding advanced considerations, the obtained stability and performance bounds may still be improved by taking advantage of the given system structure as well as by elaborating on alternative bound and norm argumentation for less conservative results.


\begin{table}[!t]
    \centering
    \caption{Numeric simulation and characteristic algorithm parameters for example from Section \ref{ssec:tvex} in comparison to \cite{MazencMN22}.}\label{tab:param}
    \vspace*{1ex}
    \renewcommand*{\arraystretch}{1.05}
    \begin{tabular}{l|c|c}\hline
        \textit{Parameter} & \textit{Symbol} & \textit{Value}  \\\hline\hline
        Mass & $M$ & $1$  \\
        Pendulum length & $L$ & $1$  \\
        Gravitation & $g$ & $9.81$  \\
        Friction function & $k(t)$ & $(1+\sin{t})/10$  \\
        Friction bound & $\bar k$ & $0.2$  \\
        Vertical acceleration & $d(t)$ & $h(t)\cos{(x_1(t))}/L$  \\
        Disturbance amplitude & $h(t)=H$ & $\{0,0.5,2,5\}$  \\
        Output coefficient & $c_o$ & $2$  \\
        Lipschitz constant & $L_\phi=\frac{g}{c_oL}$ & $4.91$  \\\hline
        Numeric sampling & $T_s$ & $1\cdot10^{\mi3}$  \\
        Output samples (As.\ref{as:sd_sys}) & $t_{i+1}-t_i$ & $0.02$ \\\hline
    \end{tabular}\\[3ex]
    \begin{tabular}{p{2.75cm}|c||c|c}\hline
        \textit{Characteristics} & \textit{Symbol} & \textit{Value} & \cite{MazencMN22}  \\\hline\hline
        Horizon length & $T=(n\!\mi\!1)\tau$ & $T=1$ & $\tau=1$ \\
        Stability factor \eqref{eq:stab_crit} & $\lambda$ & $23.9$ & $19.1$ \\
        Lin. dynamics MF (Fig.\ref{fig:umf_opt}) gain (Fig.\ref{fig:mf_bound_plot}) & $\bar\eta_a$ & 0.26 & -  \\
        Nonlinear MF gain & $\bar\eta_\phi$ & 2.76 & -  \\
        Stability cost & $J=\bar\eta_a\pl L_\phi\bar\eta_\phi$ & 13.8 & -  \\
        Maximum feasible sampling bound & $\overline{T}^\ast=1/\lambda$ & 0.042 & 0.052  \\
        Transform. gain & $\|P\|_{\cLi}$ & $0.55$ & - \\
        Pos. constant \eqref{eq:errorx} & $\alpha_x$ & $39.8$ & - \\\hline
    \end{tabular}
\end{table}


\section{Conclusion}\label{concl}

In this paper, we propose an MF-based algebraic observer combined with an output predictor for reconstructing the state of a linear time-varying system (LTV) and a sampled-data system with output nonlinearity. An $\cL2$-gain stability guarantee is provided for the LTV systems. Then, a guarantee for exponential convergence is given in the perturbation-free case as well as an input-to-state stability characterization describing the impact of measurement noise and process disturbances for the sampled-data systems.
The underlying basic MF state observer approach offers several tuning options either in the form of the moving horizon length or by using the specific disturbance gains. As a consequence of the OCF derivation, the general implementation is relatively simple with low computational cost.
In a simulation study, the developed method is compared to an established observer in \cite{MazencMN20} and \cite{MazencMN22} with a similar architecture that motivated the stability considerations of this work. The influence of different factors, such as the sampling period and perturbation due to white noise, and their relation to the stability margin are discussed.  The trajectory-based stability approach demonstrated here for the first time in the context of the MF method offers several opportunities to analyze closed-loop systems containing MF components and their design. Furthermore, the potential of the proposed algebraic observer in the event-triggered context where the sampling times $t_i$ are evenly spaced will be further investigated in future work.\\






\section*{Acknowledgment}
This project receives funding from the European Union's Horizon 2020 Research and Innovation Programme under grant agreement No 824046.

\bibliographystyle{plain}
\bibliography{biblio_KAUST,literature_MF}

\end{document}